%% file: main.tex
\documentclass[10pt,journal,twoside,compsoc]{IEEEtran}

\usepackage{amsmath}
\usepackage{graphicx}
\usepackage{txfonts}
\usepackage{multirow}
\usepackage{verbatim}
\usepackage{array}
\usepackage{subfig}
\usepackage[nocompress]{cite}
\usepackage[colorlinks, linkcolor=black, urlcolor = black, citecolor=black]{hyperref}
\hypersetup{colorlinks=true,citecolor = blue}

\usepackage{multirow}  
\newtheorem{definition}{Definition}   
\newtheorem{theorem}{Theorem}  
\newenvironment{proof}{\begin{IEEEproof}}{\end{IEEEproof}}
\usepackage{amsmath}
\usepackage{arydshln} 
\usepackage{color} 
\usepackage{algorithm}
\usepackage{algorithmic}
\renewcommand{\algorithmicrequire}{\textbf{Input:}}
\renewcommand{\algorithmicensure}{\textbf{Output:}}

\begin{document}

\title{Discovering High Utility Episodes in Sequences}

\author{
	
 	Wensheng Gan,~\IEEEmembership{Member,~IEEE,}
	Jerry Chun-Wei~Lin,~\IEEEmembership{Senior Member,~IEEE,}\\               
	Han-Chieh Chao,~\IEEEmembership{Senior Member,~IEEE,}
	and Philip S. Yu,~\IEEEmembership{Fellow,~IEEE}

	\IEEEcompsocitemizethanks{
		
		\IEEEcompsocthanksitem Wensheng Gan is with Harbin Institute of Technology (Shenzhen), Shenzhen, PR China, and with University of Illinois at Chicago, IL, USA. Email: wsgan001@gmail.com
		
		\IEEEcompsocthanksitem Jerry Chun-Wei Lin is with the Western Norway University of Applied Sciences, Bergen, Norway. Email: jerrylin@ieee.org
				
		\IEEEcompsocthanksitem Han-Chieh Chao is with the National Dong Hwa University, Hualien, Taiwan, R.O.C. Email: hcc@ndhu.edu.tw
		
		\IEEEcompsocthanksitem Philip S. Yu is with University of Illinois at Chicago, IL, USA. Email: psyu@uic.edu
	
Manuscript received September 2019.  (Corresponding author: Han-Chieh Chao)


}

}

\IEEEtitleabstractindextext{%

\begin{abstract}
 
Sequence data, e.g., complex event sequence, is more commonly seen than other types of data (e.g., transaction data) in real-world applications. For the mining task from sequence data, several problems have been formulated, such as sequential pattern mining, episode mining, and sequential rule mining.  As one of the fundamental problems, episode mining has often been studied. The common wisdom is that discovering frequent episodes is not useful enough.  In this paper, we propose an efficient utility mining approach namely UMEpi: \underline{U}tility \underline{M}ining of high-utility \underline{Epi}sodes from complex event sequence. We propose the concept of remaining utility of episode, and achieve a tighter upper bound, namely episode-weighted utilization (\textit{EWU}), which will provide better pruning. Thus, the optimized \textit{EWU}-based pruning strategies can achieve better improvements in mining efficiency.  The search space of UMEpi w.r.t. a prefix-based lexicographic sequence tree is spanned and determined recursively for mining high-utility episodes, by prefix-spanning in a depth-first way. Finally, extensive experiments on four real-life datasets demonstrate that UMEpi can discover the complete high-utility episodes from complex event sequence, while the state-of-the-art algorithms fail to return the correct results.  Furthermore, the improved variants of UMEpi significantly outperform the baseline in terms of execution time, memory consumption, and scalability.

\end{abstract}


}
\maketitle

\input{1_intro.tex}
\input{5_relatedwork.tex}

\input{2_preliminaries.tex}

\input{3_algorithm.tex}

\input{4_experiment.tex}

\input{6_conclusion.tex}

\bibliographystyle{IEEEtran}
\bibliography{main}



\vspace{-1.5cm}
\begin{IEEEbiography}[{\includegraphics[width=1in,height=1.25in,clip,keepaspectratio]{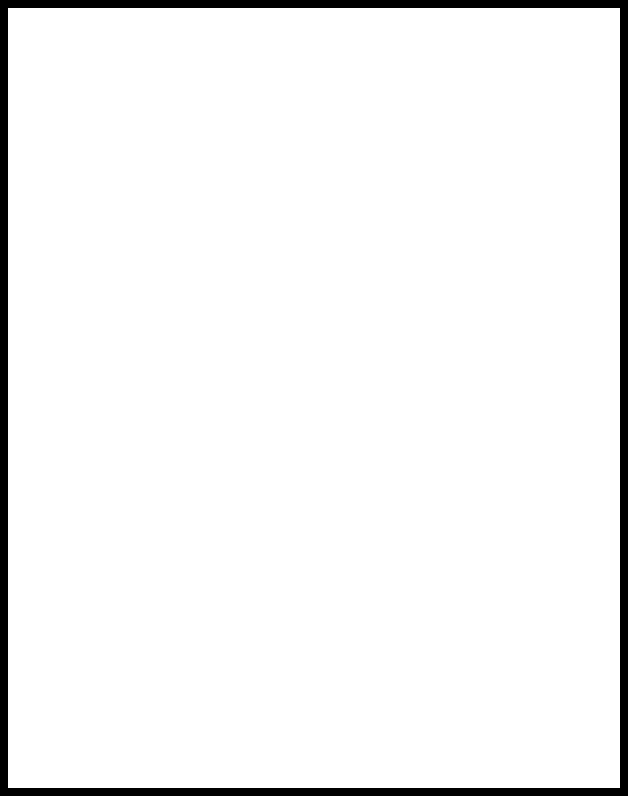}}]{Wensheng Gan (M'19)} received the Ph.D. in Computer Science and Technology, Harbin Institute of Technology (Shenzhen), Guangdong, China in 2019. He received the B.S. degree in Computer Science from South China Normal University, Guangdong, China in 2013. His research interests include data mining, utility computing, and big data analytics. He has published more than 50 research papers in peer-reviewed journals (i.e., ACM TKDD, ACM TDS, IEEE TKDE, IEEE TCYB, INS, KBS) and conferences (i.e., BigData, DSAA, DEXA, PAKDD), which have received more than 700 citations.
\end{IEEEbiography}

\vspace{-1.5cm}
\begin{IEEEbiography}[{\includegraphics[width=1in,height=1.25in,clip,keepaspectratio]{newAuthor.png}}]{Jerry Chun-Wei Lin (SM'19)}
	is an associate professor at Western Norway University of Applied Sciences, Bergen, Norway. He received the Ph.D. in Computer Science and Information Engineering, National Cheng Kung University, Tainan, Taiwan in 2010. His research interests include data mining, big data analytics, and social network. He has published more than 300 research papers in peer-reviewed international conferences and journals, which have received more than 3600 citations. He is the co-leader of the popular SPMF open-source data mining library and the Editor-in-Chief (EiC) of the \textit{Data Science and Pattern Recognition} (DSPR) journal, and Associate Editor of \textit{Journal of Internet Technology}. Dr. Lin is a Senior Member of ACM and IEEE, and a fellow of IET. 
\end{IEEEbiography}

\vspace{-1.5cm}
\begin{IEEEbiography}[{\includegraphics[width=1in,height=1.25in,clip,keepaspectratio]{newAuthor.png}}]{Han-Chieh Chao (SM'04)}
	has been the president of National Dong Hwa University since February 2016. He received M.S. and Ph.D. degrees in Electrical Engineering from Purdue University in 1989 and 1993, respectively. His research interests include high-speed networks, wireless networks, IPv6-based networks, and artificial intelligence. He has published nearly 500 peer-reviewed professional research papers. He is the Editor-in-Chief (EiC) of IET Networks and \textit{Journal of Internet Technology}. Dr. Chao has served as a guest editor for ACM MONET, IEEE JSAC, \textit{IEEE Communications Magazine}, \textit{IEEE Systems Journal}, \textit{Computer Communications}, \textit{IEEE Proceedings Communications}, \textit{Wireless Personal Communications}, and \textit{Wireless Communications \& Mobile Computing}. Dr. Chao is an IEEE Senior Member and a fellow of IET. 
\end{IEEEbiography}

\vspace{-1.5cm}
\begin{IEEEbiography}[{\includegraphics[width=1in,height=1.25in,clip,keepaspectratio]{newAuthor.png}}]{Philip S. Yu (F'93)}
	received the B.S. degree in electrical engineering from National Taiwan University, M.S. and Ph.D. degrees in electrical engineering from Stanford University, and an MBA from New York University. He is a distinguished professor of computer science with the University of Illinois at Chicago (UIC) and also holds the Wexler Chair in Information Technology at UIC. Before joining UIC, he was with IBM, where he was manager of the Software Tools and Techniques Department at the Thomas J. Watson Research Center. His research interests include data mining, data streams, databases, and privacy. He has published more than 1,300 papers in peer-reviewed journals (i.e., TKDE, TKDD, VLDBJ, ACM TIST) and conferences (KDD, ICDE, WWW, AAAI, SIGIR, ICML, etc). He holds or has applied for more than 300 U.S. patents. Dr. Yu was the Editor-in-Chief of \textit{ACM Transactions on Knowledge Discovery from Data}. He received the ACM SIGKDD 2016 Innovation Award, and the IEEE Computer Society 2013 Technical Achievement Award. Dr. Yu is a fellow of the ACM and the IEEE.
\end{IEEEbiography}

\end{document}

%% file: 1_intro.tex
\section{Introduction}

Discovering interesting patterns from various types of data (e.g., transaction data, sequence, graph, stream data, etc.) is the key problem in data mining and analytics \cite{chen1996data}. In the past decades, pattern mining \cite{agrawal1994fast,han2004mining} is well-studied for a wide range of applications. As a fundamental topic in Knowledge Discovery from Data (KDD) \cite{chen1996data}, the support-based pattern mining task has several related research subfields, including frequent itemset mining (FIM) \cite{agrawal1994fast,han2004mining} and association rule mining (ARM) \cite{samarah2011target,rashid2015mining} from transaction data, sequential pattern mining (SPM) \cite{pei2001mining} from sequential data, and frequent episode mining (FEM) \cite{mannila1997discepisodes} from a long event sequence. Up to now, many studies have been dedicated to these subfields and have arisen many important real-world applications. Note that the problem formulation of SPM and FEM are not the same although they have something in common. An episode is a non-empty ordered set of events \cite{mannila1997discepisodes}.  FIM and SPM have been well-studied, however, FEM is heavily under-developed.  One of the reasons is that FEM is more challenging and more complicated than FIM and SPM, due to the intrinsic property of utility and complexity in complex event sequences. In particular, computing statistics from a long event sequence turns out to be complex for episodes. Unfortunately, most of the developed techniques for FIM and SPM cannot be applied to address the task of frequent episode mining.

The common wisdom is that discovering frequent patterns (e.g., itemsets, sequential patterns, episodes) is not useful enough. However, most of the pattern mining algorithms in above subfields mainly adapt the co-occurrence frequency (aka support)  \cite{agrawal1994fast,han2004mining} to measure the interestingness of patterns. If the observed support of a pattern in process data is high, we assume that it deviates a lot from the expectation and this pattern is considered important. However, other implicit factors, such as the utility, interest, and risk, are not effectively utilized for evaluating the usefulness of the discovered knowledge. In some cases, frequent-based pattern mining approach might easily result to many trivial patterns, omitting the true interesting patterns \cite{2gan2018survey}. For example, in shopping behavior analysis, not only the association relationship between items/products but also the combination of products which have the most high utility can bring valuable knowledge for the retailers or managers. Inspired by the utility theory \cite{marshall2005principles}, a new mining framework namely utility-driven mining (abbreviated as \textit{utility mining} \cite{2gan2018survey}) has been introduced. Basically, the utility of a pattern refers to its importance, which can be measured in terms of utility, profit, risk, cost or other subjective measure depending on the user preference. Given a database and a user-specified minimum utility threshold (\textit{minUtil}), a pattern is called a high-utility pattern if its total utility in this database is no less than \textit{minUtil}. Utility mining has been a fruitful and active research field in data science, such as high-utility itemset mining (HUIM) \cite{liu2005two,ahmed2009efficient,tseng2013efficient,liu2012mining} and high-utility sequential pattern mining (HUSPM) \cite{alkan2015crom,wang2016efficiently,yin2012uspan}.

As mentioned before, in addition to the transaction data, there exists other types of data, such as complex event sequence, graph, and stream data. For example, an event sequence is a long sequence of events, and each event has its type and occurred time. An episode is a set of partially ordered events. Mannila \textit{et al.} \cite{mannila1997discepisodes} first introduces the frequent episode mining problem as well as two mining algorithms, WINEPI and MINEPI. Different from itemset, rule, and sequential pattern, episode is an interesting way for representing partial order relationships between events.  Complex event sequence is commonly seen and important, especially in data science. The topic of frequent episode mining (FEM) \cite{mannila1997discepisodes,ao2018discovering,ao2017mining} is different from the widely-studied frequent sequential pattern mining (SPM). For the problem of  high-utility episode mining (HUEM) in complex event sequence, it is more commonly seen in real-world, but more complicated and challenging than other tasks, e.g., HUIM and HUSPM. Discovering  high-utility episodes is difficult and challenging for two main reasons. Firstly, the search space of HUEM is quite huge than that of SPM and HUSPM, while HUEM lacks of effective pruning strategy to prune the search space. Secondly, the user-specified maximum time duration (\textit{MTD}) may easily resulting in pattern explosion. In general, the number of determined patterns will increase dramatically when \textit{MTD} increases. Ignoring the complex event sequence, the applicability of utility mining may be limited. Therefore, it is really a critical and challenging issue for discovering high-utility episodes in complex event sequences.  For utility-oriented episode mining, the pioneer work \cite{wu2013mining} formulates this task as a problem of high-utility episode mining (HUEM)  and proposes the projection-based UP-Span algorithm \cite{wu2013mining}. The goal of HUEM aims at discovering high-utility episodes which satisfy the maximum time duration (\textit{MTD}) constraint while their overall utilities are no less than a user-specified minimum utility threshold (\textit{minUtil}). Later, TSpan \cite{guo2014high}  makes use of the lexicographic sequence tree to model the HUEM task and improves the mining efficiency with two pruning strategies.  However, both UP-Span \cite{wu2013mining} and TSpan \cite{guo2014high} fail to successfully solve the HUEM problem and lead to the incorrect and incomplete mining results. Moreover, both of them  may easily encounter the mining efficiency problem and cause a lot of memory consumption. In particular, due to the intrinsic property of utility and complexity in complex event sequence data, high-utility episode mining algorithms should be fast and precise.

To this end, we propose an efficient utility-driven episode mining approach namely UMEpi: \underline{U}tility \underline{M}ining of high-utility \underline{Epi}sodes from complex event sequence. To the best of our knowledge, UMEpi not only extracts the correct high-utility episodes, but also achieves an acceptable efficiency. An interesting flexibility of UMEpi is that it can easily focus the discovering task on several real-life types of event sequences, including  serial episode,  simultaneous episode, or complex episode having both simultaneous episode and serial episode. To summarize, this paper has several contributions as follows.

\begin{itemize}
	\item We first introduce and define the concept of  the remaining utility of episode in event sequence. We then formulate an alternative definition of Episode-Weighted Utilization (\textit{EWU}), which is able to provide an accurate formulate of upper bound. UMEpi adapts the utility and sequence-order among event types as the key criterion for evaluating high-utility episodes (HUEs).
	
	\item  Flexibility of UMEpi. There is no need to make assumptions about the process event sequence is simultaneous, serial or complex. UMEpi can fast discover HUEs from any type of event sequence.
	
	\item Two optimized \textit{EWU}-based pruning strategies are developed in the UMEpi algorithm for efficiently mining high-utility episodes in depth-first way, without scanning the original sequence many times. UMEpi is able to early prune the unpromising episodes that cannot be high-utility, and does not need to perform the projection operation which costs memory  and is time-consuming. 
	
	\item Experiments on four real-world datasets have shown the effectiveness and high efficiency of the proposed UMEpi algorithm with two user-specified parameters: \textit{MTD} and \textit{minUtil}.

\end{itemize}

The rest of this paper is organized as follows. Some related works are briefly reviewed in Section \ref{sec:relatedwork}. The preliminaries and problem statement of high-utility episode mining are given in Section \ref{sec:preliminaries}.  Details of the proposed UMEpi algorithm are described in Section \ref{sec:algorithm}. The evaluation of effectiveness and efficiency of UMEpi are provided in Section \ref{sec:experiments}.  Finally, some conclusions are given in Section \ref{sec:conclusion}. 

%% file: 5_relatedwork.tex
\section{Related Work}
\label{sec:relatedwork}
To summarize, this paper is highly related to existing works in support-based episode mining, utility-based itemset/sequence mining, and utility-based episode mining.

\subsection{Support-based Episode Mining}

Discovering interesting patterns from various types of data is the key problem in the field of Knowledge Discovery from Data \cite{chen1996data}. In the past decades, a large amount of algorithms has been developed for mining interesting patterns from various types of data, such as  \cite{samarah2011target,rashid2015mining}. Most of these studies, however, mainly use support \cite{agrawal1994fast} and confidence \cite{agrawal1993mining} to discover interesting patterns, such as frequent itemsets \cite{agrawal1994fast,han2004mining,agrawal1993mining}, frequent sequential patterns \cite{pei2001mining}, frequent episodes \cite{mannila1997discepisodes,laxman2005discovering}. From itemsets that have been found frequent (or otherwise interesting), association rule \cite{agrawal1994fast} which consists of frequent itemsets has a high confidence. For an interesting pattern in association-rule mining, the occurrence of one set is regarded as a good predictor of another. A first exploration into mining frequent patterns in sequential data still presumed a transaction data \cite{agrawal1995mining}. It assumes the sequences as transactions instead of sets. Note that the data format of sequential nature  is similar but not the same to that of typical frequent itemset mining (FIM) \cite{han2004mining,agrawal1995mining}. The support-based sequential pattern mining (SPM) \cite{pei2001mining} is often different from FIM. In general, the search space of an SPM algorithm is potentially huge due to the combinatorial explosion of sequence data. Therefore, SPM is more challenging than FIM. Later, Mannila \textit{et al.}  \cite{mannila1997discepisodes} first introduced an interesting framework for mining frequent episodes in single, long sequences \cite{mannila1997discepisodes}, where events occur at certain time points. This pioneer work leads to a new research topic named frequent episode mining (FEM) \cite{mannila1997discepisodes}.

The objective of FEM is to discover episodes whose occurrences exceed a minimum support threshold, with window or time maximum duration (\textit{MTD}) constraint. As formulated, episode is an interesting way for representing partial order relationships between events. The uniqueness of an episode in an event sequence is determined by the containing events. According to the existing studies, many frequency measures have been proposed for FEM, such as the (fixed-width) windows \cite{mannila1997discepisodes}, minimal occurrences \cite{mannila1997discepisodes}, head frequency \cite{huang2008efficient},  total frequency \cite{iwanuma2004anti}, and non-overlapped occurrences \cite{laxman2005discovering}. Ao \textit{et al.} \cite{ao2017mining} addressed the new problem of mining precise positioning episode rules. In literature, most of the FEM algorithms use the minimal occurrence  \cite{meger2004constraint} to measure the frequencies of interesting episodes. Despite of the vast amounts of research efforts in support-based episode mining, fewer studies of FEM consider the utility factor into account. For example, when performing a market basket analysis on retail data, many frequent but not profitable episodes may be found. In the risk perdition task, the rare or frequent episodes cannot successfully capture the high risk events.

\subsection{Utility-based Itemset/Sequence Mining}

To address the early mentioned problems, utility-oriented pattern mining (also called utility mining) \cite{tseng2013efficient,yao2006mining,gan2018survey} is proposed as a new data mining framework. Utility mining aims at identifying the high-utility patterns but not the frequent ones.  The occur quantity and unit utility of objects/items, as well as other implicit factor are taken into account. Up to now, the problem of utility mining has been extensively studied, and leads to many related research subfields, including high-utility itemset mining (HUIM) \cite{liu2005two,yun2017efficient}, high-utility sequential pattern mining (HUSPM) \cite{yin2012uspan}, and high-utility rule mining (HURM) \cite{mai2017lattice}.

For the task of HUIM from itemset-based data,  many algorithms have been developed,  such as Two-Phase \cite{liu2005two}, IHUP \cite{ahmed2009efficient}, UP-growth \cite{tseng2010up}, UP-growth+ \cite{tseng2013efficient}, and HUI-Miner \cite{liu2012mining}. To be more specific, these existing HUIM algorithms can be mainly classified into the following categories: Apriori-like, tree-based, utility-list-based, and hybrid approaches, as reviewed in \cite{2gan2018survey}. The Two-Phase \cite{liu2005two} algorithm is the early Apriori-like approach for HUIM, and it first introduced the transaction-weighted utilization (\textit{TWU}) concept \cite{liu2005two}. Later, some tree-based HUIM algorithms, e.g., IHUP \cite{ahmed2009efficient}, HUP-tree \cite{huptree}, UP-growth \cite{tseng2010up}, UP-growth+ \cite{tseng2013efficient}, are developed and all outperform the Apriori-like ones. Liu \textit{et al.} \cite{liu2012mining} then introduced the utility-list structure that utilizes a concept namely remaining utility. Towards a better mining efficiency, several more efficient approaches have been developed for discovering high-utility itemsets, such as FHM \cite{fournier2014fhm}, HUP-Miner \cite{krishnamoorthy2015pruning}, and EFIM \cite{zida2017efim}. In addition to efficiency, the effectiveness issue of utility mining has been extensively studied, for example, several interesting works are reported in \cite{gan2018survey,nguyen2019mining,yun2017efficient,yun2018damped,tseng2016efficient}. An up-to-date overview of the current development of utility mining can be referred to \cite{2gan2018survey}.

For the task of HUSPM \cite{alkan2015crom,wang2016efficiently,yin2012uspan} from sequence-based data, it uses some special data structures and upper bound on utility of sequence. The early Apriori-like approach for HUSPM is the US \cite{ahmed2010mining} and UL \cite{ahmed2010mining} algorithms which utilize the Sequence-Weighted Utilization (\textit{SWU}) model \cite{ahmed2010mining,yin2012uspan}. Yin \textit{et al.} \cite{yin2012uspan} presented a generic definition of the HUSP mining framework and proposed a new mining algorithm named USpan.  Alkan \textit{et al.} proposed HuspExt \cite{alkan2015crom} with a Cumulate Rest of Match (CRoM) based pruning technique.  The HUS-Span algorithm \cite{wang2016efficiently} introduces a data structure namely utility-chain and utilizes prefix extension utility (\textit{PEU}) as an upper bound. However, HUS-Span is not efficient enough since the generate-and-test approach creates an overflow of candidate sequences. Recently, two more efficient algorithm, such as ProUM \cite{gan2019proum} and HUSP-ULL \cite{gan2019fast}, were developed to fast identify high-utility sequential patterns. ProUM utilizes the utility-array structure and \textit{PEU} \cite{gan2019proum} upper bound, and HUSP-ULL utilizes the utility-linked list \cite{gan2019fast} structure and two powerful pruning strategies.  More current development of HUSPM can be referred to in literature reviews \cite{2gan2018survey,gan2018privacy}.

\subsection{Utility-based Episode Mining}

Support constraint has been a popular measure for extracting episodes. Although there are many studies on frequent episode mining and high-utility itemset/sequence mining, utility-driven episode mining has been studied less. In order to discovering high-utility episodes rather than those frequent ones, Wu \textit{et al.} \cite{wu2013mining} first proposed an interesting concept called Episode-Weighted Utilization (\textit{EWU}) and introduced the problem statement of high-utility episode mining (HUEM). To find interesting episodes whose utilities exceed the expected utility value, the UP-Span algorithm \cite{wu2013mining} and two pruning strategies were further presented as the pioneer work. The Episode-Weighted Utilization (\textit{EWU}) in HUEM  is equivalent to the Sequence-Weighted Utilization (\textit{SWU}) in HUSPM. However, the HUEM problem is different from HUSPM problem, since the former is considered more complicated and challenging than the later.  Most of the developed techniques for FEM, HUIM, and HUSPM cannot directly be applied to HUEM. As an enhanced algorithm of UP-Span, TSpan \cite{guo2014high} was proposed to discover high-utility episodes (HUEs) using two tighter upper bounds, which can reduce the search space over the prefix tree.  Consider the rule generation, Lin \textit{et al.} \cite{lin2015discovering} studied the problem of  utility-based episode rule mining in complex event sequences. However, both UP-Span and TSpan suffer from several performance drawbacks: 1) The concept of \textit{EWU} and pruning strategies are only based on an approximate upper bound. 2) Their results contain several errors, e.g., incorrect and incomplete. 3) Their mining efficiency in terms of running time, memory consumption and scalability might not efficient enough to deal with a long event sequence.  Therefore, we should develop the fast and precise algorithms for discovering high-utility episodes.

%% file: 2_preliminaries.tex

\section{Preliminaries and Problem Formulation}
\label{sec:preliminaries}

Based on the previous studies \cite{mannila1997discepisodes,wu2013mining,laxman2005discovering}, some concepts and principles of high-utility episode mining are introduced firstly.

\subsection{Preliminaries of Utility Mining on Event Sequence}

An \textit{event} is defined as a pair ($e, T_i$) where $e$ is the event type and $T_i \in N^+$ is the \textit{occur time} when this event happens. A partial ordered collection of events is called an \textit{episode}.  An \emph{event sequence} is defined as an ordered sequence of simultaneous event sets, such that $S$ = $<$\{($s_1, T_1$), ($s_2, T_2$), $ \dots$, ($s_n, T_n$)\}$>$ where each simultaneous event set $s_i \in S$ consists of several simultaneous or serial events. Each simultaneous event set is associated with a unique time point $T_i$ ($T_i < T_j$), for all 1 $\leq i < j \leq n$.  Let $e$ be a distinct event of $E$ = $\{e_1,$ $e_2,$ $\dots, e_m\}$. Each event $e \in E$ in event sequence $S$ is associated with a unique positive value $pr(e)$ namely its \emph{unit utility} (also called \textit{external utility}). For each event $e$ in $T_i$, a positive number $q(e, T_i)$ is called its \emph{occur quantity} (also called \textit{internal utility}). An $l$-episode means that the length of this episode is $l$.

\begin{figure}[htbp]
	\centering 
	\includegraphics[scale=1.1]{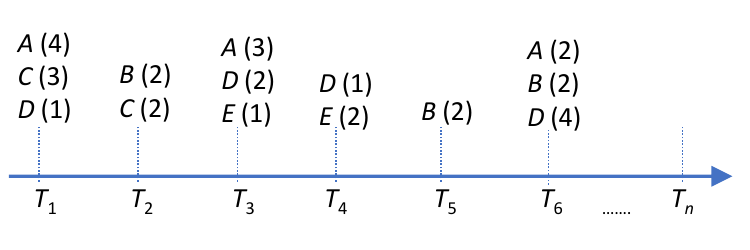}
	\captionsetup{justification=centering}
	\caption{A complex event sequence with occur quantity.}
	\label{fig:data}	
\end{figure}

\begin{table}
	\centering
	\caption{External unit utility value}
	\begin{tabular}{|c|c|c|c|c|c|} \hline
		\textbf{Event} & $A$ & $B$ & $C$ & $D$ & $E$ \\ \hline
		\textbf{Utility} (\$) & 1 & 5 & 2 & 3 & 7 \\ \hline
	\end{tabular}
	\label{figUnitProfits}
\end{table}

\begin{definition}[Complex episode with simultaneous and serial concatenations]
	\rm  Let $\alpha$ = $<$$(E_1)$, $(E_2)$, $\ldots$, $(E_x)$$>$ and $\beta$ = $<$$(E_1')$, $(E_2')$, $\ldots$, $(E_y')$$>$ be two different episodes. The simultaneous concatenation of $\alpha$ and $\beta$ is defined as \textit{Simult-Concatenate}($\alpha$, $\beta$) = $<$$(E_1)$, $(E_2)$, $\ldots$, $(E_x \sqcup E_1')$, $(E_2')$, $\ldots$, $(E_y')$$>$. The serial concatenation of $\alpha$ and $\beta$ is defined as \textit{Serial-Concatenate}($\alpha$, $\beta$) = $<$$(E_1)$, $(E_2)$, $\ldots$, $(E_x)$, $(E_1')$, $(E_2')$, $\ldots$, $(E_y')$$>$. Then the simultaneous episodes are related to simultaneous concatenation, and the serial episodes are related to serial concatenation. In general, a complex episode always has both simultaneous episode and serial episode. 
\end{definition}

For example, Fig. \ref{fig:data} shows a complex event sequence $S$ = $<$(\{$A$ (4), $C$ (3)\}, $T_1$), (\{$B$ (2), $C$ (2)\}, $T_2$), (\{$A$ (3), $D$ (2), $E$ (1)\}, $T_3$), (\{$D$ (1), $E$ (2)\}, $T_4$), (\{$B$ (2)\}, $T_5$), (\{$A$ (2), $D$ (4)\}, $T_6$)$>$. Among this sequence, $<$$A$ (4)$>$ and $<$$C$ (3)$>$ are the simultaneous episodes in $T_1$, while $<$$A$ (4) $\rightarrow$ $C$ (2)$>$ are the serial episodes. Note that the simultaneous episode does not consider the sequence-order inside itself, e.g., $<$\{$A$ (4), $C$ (3)\}$>$ and $<$\{$C$ (3), $A$ (4)\}$>$ are the same episode. However, $<$$A \rightarrow C$$>$ and $<$$C \rightarrow A$$>$ are two different serial episodes since their sequence-orders are not the same.

\begin{definition}[Occurrence] 
	\rm Given an episode $\alpha$ = $<$$(E_1)$, $(E_2)$, $\ldots$, $(E_k)$$>$, the time point interval $[T_s, T_e]$ is called an occurrence of $\alpha$ iff: (1) $\alpha$ occurs in $[T_s, T_e]$, and (2) the first simultaneous event set in $\alpha$ occurs at  $T_s$ (called \textit{start time}) and the last simultaneous event set $E_k$ of $\alpha$ occurs at $T_e$ (called \textit{end time}). The set of all occurrences of $\alpha$ in $S$ is denoted as $occSet(\alpha)$. 
\end{definition}

In this paper, note that a complex episode may have either simultaneous/serial episode or two types of episode. Checking whether an episode occurs in an event sequence is an NP-complete problem \cite{tatti2011mining}. For example, the set of all the occurrences of the episode $<$$\{A, D\} \rightarrow B$$>$ in Fig. \ref{fig:data} is $occSet$($<$$\{A, D\} \rightarrow B$$>$) = \{$[T_1, T_2]$, $[T_1, T_5]$, $[T_1, T_6]$, $[T_3, T_5]$, $[T_3, T_6]$\}. Up to now, several measures of occurrences (e.g., minimal occurrence \cite{mannila1997discepisodes}, non-overlapped occurrence \cite{laxman2007nonoverlapepisodes}) have been introduced for frequent episode mining.

\begin{definition}[Minimal occurrence] 
	\rm Given two  time point intervals $[T_s, T_e]$ and $[T_s', T_e']$ of occurrences of an episode $\alpha$, $[T_s', T_e']$ is called the time point sub-interval of $[T_s, T_e]$ if $T_s \leq T_s'$ and $T_e' \leq T_e$.  $[T_s, T_e]$ is called a minimal occurrence of episode $\alpha$ if the following conditions are satisfied:  (1) $[T_s, T_e]$ is an occurrence of $\alpha$, and (2) there is none occurrence $[T_s', T_e']$ of $\alpha$ such that $[T_s', T_e']$ is a sub-interval of $[T_s, T_e]$. To summarize, minimal occurrence is a kind of occurrence of episode which cannot contain any other occurrence of same episode. Thus, for an episode, there is only one precise minimal occurrence with respect to each occurrence. Let $mo$($\alpha$) denote the minimal occurrences of $\alpha$ w.r.t. $[T_s, T_e]$, then the complete set of minimal occurrences of $\alpha$ in $S$ is $moSet$($\alpha$).

\end{definition}

For example, in Fig. \ref{fig:data}, the time point interval $[T_1, T_2]$ is a minimal occurrence of the episode $<$$\{A, D\} \rightarrow B$$>$, and $moSet$($<$$\{A, D\} \rightarrow B$$>$) = \{$[T_1, T_2]$, $[T_3, T_5]$\}. Thus, we can see that the minimal occurrence is the shortest time interval that contains a particular episode.

\begin{definition}[Maximum time duration]
	\rm The maximum time duration is a user-specified constraint for the discovered episodes. An episode $\alpha$ satisfies the maximum time duration constraint iff $(T_e - T_s) \leq MTD $\footnote{Note that both UP-Span and TSpan use $(T_e - T_s + 1) \leq MTD $, but we use $(T_e - T_s) \leq MTD $ for easier understanding and more precise.}. In this paper, we assume that $mo(\alpha)$ should always satisfy the maximum time duration constraint.

\end{definition}

\begin{definition}[Sub-episode and super-episode] 
	\rm Given two episodes $\alpha$ = $<$$E_1, E_2, \ldots, E_n$$>$ and $\beta $ = $<$$E_1', E_2', \ldots, E_m'$$>$ where $m \leq n$, the episode $\beta $ is called a \textit{sub-episode} of $\alpha$ iff there exists $m$ integers $1 \leq i_1 < i_2 < \ldots < i_m \leq n$ such that $E_{k}' \in E$ for $ 1 \leq k \leq m \leq n $. In this case, conversely, episode $\alpha$ is called the \textit{super-episode} of $\beta $.
\end{definition}

For example, the episode $<$$\{A, D\} \rightarrow B$$>$ is called a \textit{super-episode} of $<$$A$$>$, $<$$\{A, D\}$$>$, $<$$A \rightarrow B$$>$, and $<$$D \rightarrow B$$>$.

\begin{definition}[Utility of a simultaneous event set w.r.t. time point]
	\rm The utility of a simultaneous event set $E \subseteq S$ in a time point $T_i$ is $ u(E, T_i)$ = $\sum_{e \in E} {u(e, T_i)}$, in which $u(e, T_i)$ is the utility of an event $e \in E$ in this time point $T_i$, such as $u(e, T_i)$ = $pr(e) \times q(e, T_i)$. Hence, $ u(E, T_i)$  represents the all utilities generated by all events $e \in E$ in each time point $T_i$. 
\end{definition}

In Fig. \ref{fig:data}, the utility of an event \{$A$\} in time point $T_3$ is $u(A, T_3)$ = $3 \times \$1$ = \$3, and the utility of event set $\{A, D, E\}$ in $T_3$ is  $u(\{A, D, E\}, T_3)$ =  $u(A, T_3)$ + $u(D, T_3)$ + $u(E, T_3)$ = 3 $\times$ \$1 + 2 $\times $ \$3  + 1 $\times $ \$7 = \$16.

\begin{definition}[Total utility of an event sequence]
	\rm Given a complex event sequence $S$, the total utility of all events in a time point $ T_i $, denoted as $tu(T_i)$, is defined as $tu(T_i)$ = $ \sum_{e_{j}\in T_i}u(e_{j}, T_i) $, where $e_j$ is the $j$-th event in $ T_i $. Then the total utility of $S$, denoted as $TU$, is defined as: $TU$ = $\sum_{T_i \in S}tu(T_i) $.
\end{definition}

Consider the first time point in Fig. \ref{fig:data}, $ tu(T_{1}) $ = $ u(A, T_{1}) $ + $ u(C, T_{1}) $ + $ u(D, T_{1}) $ = \$4 + \$6 + \$3 = \$13. Then the utilities of \textit{T}$ _{1} $ to \textit{T}$ _{6} $ can be  calculated as \textit{tu}($ T_{1} $) = \$13, \textit{tu}($ T_{2} $) = \$14, \textit{tu}($T_{3} $) = \$16, \textit{tu}($T_{4}$) = \$17, \textit{tu}($T_{5}$) = \$10, and \textit{tu}($T_{6}$) = \$24, respectively. Thus, the total utility in $S$ is $TU$ = \$13 + \$14 + \$16 + \$17 + \$10 + \$24 = \$94. To determine the utility of an episode, we need to firstly define how we count its occurrences. In the task of HUEM, we use the minimal occurrence to calculate the utility, which is defined below.

\begin{definition}[Utility of an episode w.r.t. a minimal occurrence]
	\rm Let $mo(\alpha)$ = $[T_s, T_e]$ be the minimal occurrence of the episode $\alpha$ = $<$$(E_1), (E_2), \ldots, (E_k)$$>$, where each simultaneous event set $E_i \in \alpha $ is associated with a time point $T_i$. Under the definition of minimal occurrence, the utility of the episode $\alpha$ w.r.t. $ mo(\alpha)$ can be defined as $u(\alpha, mo(\alpha))$ = $\sum _{i=1}^{k}u(E_i, T_i)$, $E_i \subseteq \alpha$ and $T_s \leq T_i \leq T_e$. 
\end{definition}

Referring to the previous example, assume $\alpha$ is $<$$\{A, D\} \rightarrow B$$>$. Due to $moSet$($<$$\{A, D\} \rightarrow B$$>$) = \{$[T_1, T_2]$, $[T_3, T_5]$\}, we have $u(\alpha, mo_{1}(\alpha))$ = $u(\{A, D\}, T_1)$ + $u(B, T_2)$ = \$7 + \$10 = \$17, and $u(\alpha, mo_{2}(\alpha))$ = $u(\{A, D\}, T_3)$ + $u(B, T_5)$ = \$9 + \$10 = \$19.

\begin{definition}[Utility of an episode w.r.t. an event sequence]
	\rm  Consider the entire event sequence $S$, let $u(\alpha)$ denote the total utility of an episode $\alpha$ in $S$ and $moSet$($\alpha$) = [$mo_{1}(\alpha), mo_{2}(\alpha), \ldots, mo_{n}(\alpha)$] denote the complete set of minimal occurrences of $\alpha$ in $S$, then $ u(\alpha)$ = $\sum_{i=1}^{n} u(\alpha, mo_{i}(\alpha)) $.  
\end{definition}

\begin{definition}[High-utility episode]
	\rm Given a complex event sequence $S$, the maximum time duration \textit{MTD}, and a user-specified minimum utility threshold \textit{minUtil}, an episode is said as a high-utility episode (abbreviated as HUE), iff its total utility in $S$ satisfies $u(\alpha) \geq minUtil \times TU$. Otherwise, the episode is called a low-utility episode (abbreviated as LUE). 
\end{definition}

In Fig. \ref{fig:data}, the utility of the episode $<$$\{A, D\} \rightarrow B$$>$ is $u$($<$$\{A, D\} \rightarrow B$$>$) = \$17 + \$19 = \$36. If we assume \textit{MTD} = 2 and \textit{minUtil} = 50\%, the episode $<$$\{A, D\} \rightarrow B$$>$ is less than 50\% $\times$ \$94 (= \$47), thus it is a low-utility episode. However, the episode $<$$\{D, E\} \rightarrow B \rightarrow \{A, B, D\}$$>$, which utility is \$51, is a high-utility episode. When \textit{MTD} = 2 and \textit{minUtil} = 50\%, the complete set of HUEs in Fig. \ref{fig:data} is shown in Table \ref{table:HUEs}. Notice that the utility of $<$$\{B, C\} \rightarrow \{A, D, E\} \rightarrow \{D, E\}$$>$   calculated by UP-Span is equal to \$43, which is incorrect.

\begin{table}
	\centering
	\caption{Final high-utility episodes in the running example (the true utility is denoted as \textit{Utility}, and the results of \textit{Utility}* are derived by UP-Span)}
	\begin{tabular}{|c|c|c|} \hline
		\textbf{Episode}   &   \textbf{\textit{Utility}}    &   \textbf{\textit{Utility}*}     \\ \hline
		$<$$\{B, C\} \rightarrow \{A, D, E\} \rightarrow \{D, E\}$$>$ &  \$47   &  \$43   \\ \hline
		$<$$D  \rightarrow B \rightarrow \{A, D\}$$>$ &  \$49    &  \$49     \\ \hline
		$<$$\{D, E\} \rightarrow B \rightarrow \{B, D\} $$>$ &  \$49     &  \$49    \\ \hline
		$<$$\{D, E\} \rightarrow B \rightarrow \{A, B, D\} $$>$ &  \$51     &  \$51    \\ \hline
		$<$$E \rightarrow B \rightarrow \{A, B, D\} $$>$ &  \$48    &  \$48     \\ \hline
	\end{tabular}
	\label{table:HUEs}
\end{table}

\subsection{Problem Formulation}
Typically, frequent or high-utility episode mining is carried out from one long event sequence which consists of a large amount of  time-stamped events. Based on the above definitions, then we have the following problem formulation.
\begin{definition}
	\label{def_HUEM}
	\rm Given a complex event sequence $S$ with simultaneous or series events that having external and internal utility of events,  the maximum time duration (\textit{MTD}) as constrain, and a user-specified minimum utility threshold (\textit{minUtil}), the problem of high-utility episode mining (HUEM) aims at discovering all the episodes whose utilities are no less than $minUtil \times TU$ with the \textit{MTD}  constraint. 
\end{definition}

For the task of HUEM, there are two user-specified parameters: \textit{MTD} and \textit{minUtil}. Here an episode (also known as serial episode, or complex episode \cite{mannila1997discepisodes}) refers to a totally ordered set of events. To summarize, the addressed task of HUEM aims at identifying all high-utility episodes which overall utilities exceed an expected threshold ($minUtil \times TU$)  when dealing with a complex event sequence which consists of a large set of events. Discovering high-utility episodes is a good way to unearth utility-driven information and knowledge in the sequence data.

%% file: 3_algorithm.tex
\section{Proposed UMEpi Algorithm}
\label{sec:algorithm}

In this section, we propose an efficient UMEpi algorithm to discover high-utility episodes that satisfy the constraints of \textit{MTD} and \textit{minUtil}. UMEpi discovers HUEs by spanning the search space w.r.t. an conceptual lexicographic sequence (LS)-tree. Moreover, the remaining utility of episode and a tight upper bound namely episode-weighted utilization (\textit{EWU}) are developed and utilized in the pruning strategies. Details of the downward closure property of \textit{EWU}, the pruning strategies with optimized \textit{EWU}, and the main procedures of UMEpi are described below, respectively.

\subsection{Downward Closure Property}

\begin{definition}
	\rm  ($I$-Concatenation and $S$-Concatenation). For an $l$-episode $\alpha$, when an event is appended to the end of $\alpha$, it will construct an extended  ($l$+1)-episode, which is a 1-extension of episode $\alpha$. This 1-extension operation is called \textit{concatenation}. To be more specific, if the duration time of the new extended episode is the same as that of $\alpha$, this operation is an $I$-Concatenation. However, if the time duration of the new extended episode is that of $\alpha$ increased by 1, this operation is called an $S$-Concatenation. 

\end{definition}

For example, consider an 2-episode $<$$\{B, C\}$$>$, the concatenation of $<$$\{B, C, A\}$$>$ is its $I$-Concatenation, while  $<$$\{B, C\}  \rightarrow \{A\}$$>$ is the result of an $S$-Concatenation. Note that $<$$\{B, C, A, E\}$$>$ and $<$$\{B, C\}  \rightarrow \{A, D, E\}$$>$ are also the extended episodes of $<$$\{B, C\}$$>$, which are called the 2-extension and 3-extension, respectively.

\begin{definition}[Lexicographic sequence tree]
	\rm  In the lexicographic sequence tree \cite{ayres2002sequential}, a) the root node\footnote{Without ambiguity, the terms episode and node will be used interchangeably in this paper.} of the prefix-based tree is empty; b) for a parent node (also called prefix node), all the child nodes are generated by the $I$-Concatenation or $S$-Concatenation; and c) all the child nodes of a prefix node are listed in a specific order (e.g., incremental order, arbitrary order, or lexicographic order).
\end{definition}

Given a long sequence $S$, the lexicographic sequence (LS)-tree is a structure that captures all the possible episodes. Specifically, the total number of all the possible episodes in this search space is extremely huge \cite{laxman2005discovering,tatti2011mining}. Obviously, the brute-force mechanism (e.g., enumerate and then determine all the possible episodes) is not an acceptable solution. Since many datasets always contain enormous piles of low-utility patterns, selecting only the high-utility ones is not easy. Most of the existing studies have been demonstrated that the utility measure \cite{liu2005two} are neither monotonic nor anti-monotonic. Based on the definition of HUE, the HUE does not hold anti-monotonicity. In other words, a HUE may have a lower, equal or higher utility than any of its sub-episodes.  Without holding the anti-monotonicity, it is hard to efficiently reduce the search space of the addressed problem of HUEM.

For HUEM, Episode-Weighted Utilization (\textit{EWU}) \cite{wu2013mining} was proposed as a upper bound on utility of an episode in a complex event sequence.  Although these approaches achieve a speed-up of several orders of magnitude over the brute-force algorithm, their common drawback is that they all use an incorrect \textit{EWU} value, and output the incorrect results. Thus, the \textit{EWU} value in existing algorithms is not true and may cause the incorrect results. Therefore, the previous works on HUEM can not really extract the complete HUEs.

\begin{definition}[Episode-Weighted Utilization \cite{wu2013mining}]
	\rm Let $mo(\alpha)$ = $[T_s, T_e]$ be a minimal occurrence of the episode $\alpha$ = $<$$(E_1$), $(E_2)$, $\ldots$, $(E_{k-1})$, $(E_k)$$>$, where each simultaneous event set $E_i \in \alpha $ is associated with a time point $T_i$ ($1 \leq i \leq k$) and $mo(\alpha)$ satisfies \textit{MTD}. The episode-weighted utilization of $\alpha$ w.r.t. $mo(\alpha)$ is defined as: \textit{EWU}($\alpha$, $mo(\alpha)$) =  $\sum _{i=1}^{k-1}u(E_i, T_i)$ + $\sum _{i=T_{e}}^{T_{s} + MTD}u(SE_i, T_i)$ in \cite{wu2013mining}, or \textit{EWU}($\alpha$, $mo(\alpha)$) =  $\sum _{i=1}^{k}u(E_i, T_i)$ + $\sum _{i=T_{e}}^{T_{s} + MTD}u(SE_i, T_i)$ in \cite{guo2014high}, where $SE_i$ is the simultaneous event set at the time point $T_i$ in $S$. Thus, for  $\alpha$ in $S$, we have the accumulative  \textit{EWU} w.r.t. $moSet$($\alpha$) = [$mo_{1}(\alpha)$, $mo_{2}(\alpha)$, $\ldots$, $mo_{n}(\alpha)$], such that  \textit{EWU}($\alpha$) = $\sum _{i=1}^{n}EWU(\alpha, mo_{i}(\alpha))$ \cite{wu2013mining}.
	
\end{definition}

\begin{definition}[Promising episode and unpromising episode]
	\rm An episode $\alpha$ is called a promising episode iff \textit{EWU}$(\alpha) \geq minUtil \times TU$. Otherwise it is an unpromising episode. 
\end{definition}

To be more specific, for a given \textit{MTD} and $[T_s, T_e]$ of $\alpha$, then the potential minimal occurrence $mo(\alpha)$ only occurs at interval [$T_e, T_s$ + \textit{MTD}]. For example, assume \textit{MTD} = 2 in Fig. \ref{fig:data}, the \textit{EWU} of $\alpha$ = $<$$\{A, D\} \rightarrow B$$>$  w.r.t. its minimal occurrence $[T_1, T_2]$ is calculated as  $EWU$($<$$\{A, D\} \rightarrow B$$>$, $[T_1, T_2]$) = $u(\{A, D\}$, $T_1)$ + $u(\{B, C\}$, $T_2)$ + $u(\{A, D, E\}$, $T_3)$ = \$7 + \$14 + \$16 = \$37. However, with the definition in \cite{guo2014high}, the results is $EWU$($<$$\{A, D\} \rightarrow B$$>$, $[T_1, T_2]$) = $u(\{A, D\}$, $T_1)$ + $u(B$, $T_2)$ + $u(\{B, C\}$, $T_2)$ + $u(\{A, D, E\}$, $T_3)$ = \$7 + \$10 + \$14 + \$16 = \$47.

The \textit{EWU} \cite{wu2013mining} is an upper bound on utility measure by overestimating the overall utility of an episode in entire event sequence, but avoid missing any high-utility episodes.  This is justified by the following theorem as in \cite{wu2013mining}. However, a large amount of low-utility episodes still may be regarded as candidates since \textit{EWU} is a loose upper-bound.

\begin{definition}[High Weighted Utilization Episode \cite{wu2013mining}]
   \rm  Given a event sequence $S$, an episode is called High Weighted Utilization Episode (abbreviated as HWUE) in $S$  iff its \textit{EWU} is no less than \textit{minUtil} $\times$ $TU$. 
\end{definition}

\begin{theorem}[Episode-Weighted Downward Closure property]
	\label{theorem:EWU}
	\rm Let $\alpha$ and $\beta$ be two episodes, and $\gamma$ is a super-episode of $\alpha$ and $\beta$, either generated by \textit{Simult-Concatenate}($\alpha$, $\beta$) or generated by \textit{Serial-Concatenate}($\alpha$, $\beta$). The Episode-Weighted Downward Closure (abbreviated as EWDC) property means that if $EWU(\alpha) <$ \textit{minUtil} $\times $ $ TU$ or $EWU(\beta) <$ \textit{minUtil} $\times $ $ TU$, $\gamma$ is a low utility episode. It is important to note that \textit{EWU} is the upper bound of the episode as prefix when performing prefix-spanning. In other words, the non-HWUE may still be the sub-episode  of the final HUEs (as the suffix in HUEs). Thus, we can not remove those episodes which are not HWUEs when prefix-spanning for discover HUEs.  
\end{theorem}
\begin{proof}
	Let $moSet$($\alpha$) = [$mo_{1}(\alpha)$, $mo_{2}(\alpha)$, $\ldots$, $mo_{x}(\alpha)$], and $moSet$($\gamma$) = [$mo_{1}'(\gamma)$, $mo_{2}'(\gamma)$, $\ldots$, $mo_{y}'(\gamma)$]. Because $\gamma$ = \textit{Simult-Concatenate}($\alpha$, $\beta$) or \textit{Serial-Concatenate}($\alpha$, $\beta$), the following relationship holds: $|moSet(\alpha)|\geq|moSet(\gamma)|$ \cite{ma2004finding}. According to the definition of \textit{EWU}, $EWU(\alpha)$ $ \geq $ $EWU(\gamma)$. Thus, if $EWU(\alpha) <$  \textit{minUtil} $\times$ $ TU$, we have $u(\gamma) \leq EWU(\gamma) \leq EWU(\alpha) <$ \textit{minUtil} $\times$ $TU$, which yields that $\gamma$ is low utility.
\end{proof}

Based on the above definitions for HUEM, we can obtain two important observations from the existing studies.

\textit{\textbf{Observation 1}}. Similar to the concept of \textit{TWU} and \textit{SWU}, \textit{EWU} serves as upper bound of an episode's utility and maintains the downward closure property. It is important to notice that both \textit{TWU} and \textit{SWU} have the global downward closure property, while \textit{EWU} does not guarantee this. In the running example, the episode  $\{A, B, D\} $ is not the HWUE since \textit{EWU}$(\{A, B, D\})$ = \$24, but its super-episodes $<$$\{D, E\} \rightarrow B \rightarrow \{A, B, D\} $$>$ and $<$$E \rightarrow B \rightarrow \{A, B, D\} $$>$ are the final HUEs under the above settings, as shown in Fig. \ref{table:HUEs}. Thus, the \textit{EWU} upper bound only has the local downward closure property in subtrees of LS-tree.

\textit{\textbf{Observation 2}}. In UP-Span, the strategy namely \underline{D}iscarding \underline{G}lobal unpromising \underline{E}vents (DGE) \cite{wu2013mining} is applied to remove the global unpromising 1-episodes (which \textit{EWU} are less than $minUtil \times TU$) from the complex event sequence. As mentioned earlier, this strategy cannot discover the complete set of HUEs and will lead to false results. For instance, if an 1-episode is not the HWUE, it still would be the part (sub-episode) of final HUEs. Under the same parameter settings, only episode $<$$D  \rightarrow B \rightarrow \{A, D\}$$>$ = \$49 is returned by UP-Span as the final HUEs in Fig. \ref{table:HUEs}. In the same way, TSpan also fails to identify the complete HUEs.

\subsection{Pruning Strategies for Searching HUEs}

Finding all high-utility episodes without a powerful pruning strategy would be infeasible due to the prohibitively large number of candidate patterns. As described in the definition of \textit{EWU}, the computation of $EWU(\alpha)$ w.r.t. one minimal occurrence $mo(\alpha)$ = [$T_s$, $T_e$], in fact, consists of two parts, an episode $\alpha$ its own utility w.r.t. $mo(\alpha)$, and the remaining utility of $\alpha$ starting from time point $T_e$ to $T_s$ + \textit{MTD}. Let $moSet(\alpha)$ denote the set of the minimal occurrences of $\alpha$ in $S$, then we can calculate $EWU(\alpha)$ by accumulating all the \textit{EWU} value in each minimal occurrence, such that \textit{EWU}($\alpha$) = $\sum _{i=1}^{n}EWU(\alpha, mo_{n}(\alpha))$. However, the first and second part have some overlaps with episode $\alpha$ during the computation. For example, in Fig. \ref{fig:data}, $EWU$($<$$\{A, D\} \rightarrow B$$>$, $[T_1, T_2]$) = $u(\{A, D\}$, $T_1)$ + $u(B$, $T_2)$ + $u(\{B, C\}$, $T_2)$ + $u(\{A, D, E\}$, $T_3)$ = \$7 + \$10 + \$14 + \$16 = \$47, thus the event $B$ in $T_2$ is overlap. Therefore, the estimated utility upper bound is still loose. To speed up the mining performance, we propose two optimized strategies to reduce the upper bound namely \textit{EWU}. Details are the optimized \textit{EWU} and corresponding pruning strategies are described below. More formally, to compute the optimized \textit{EWU} of the episodes, the concept of remaining utility of episode is defined below.

\begin{definition}[Remaining utility of an episode]
	\rm At a specific time point $T_i$, the remaining utility of an episode $\alpha$ in $T_i$ is the accumulative utilities of all events after this episode in $T_i$: $ru(\alpha, T_i)$ = $\sum_{e'\notin \alpha \wedge \alpha \prec e'}u(e', T_i)$.
	
\end{definition}

Basically, the remaining utility of an episode in $T_i$ means the sum of the utilities after this episode in $T_i$. For example, in $T_3$ of the running example,  $ru(A, T_3)$ = $u(D, T_3)$ + $u(E, T_3)$ = \$13, and $ru(D, T_3)$ = $u(E, T_3)$ = \$7.

\begin{definition}[Optimized Episode-Weighted Utilization (optimized \textit{EWU} version 1.0 is denoted as $EWU_{opt}$]
	\rm Assume an episode $\alpha$ = $<$$(E_1$), $(E_2)$, $\ldots$, $(E_{k-1})$, $(E_k)$$>$ with one of its minimal occurrences as $mo(\alpha)$, the optimized episode-weighted utilization of $\alpha$ w.r.t. $mo(\alpha)$ is defined below by removing the overlap utilities in $T_{e}$:  \textit{EWU}($\alpha$, $mo(\alpha)$) =  $\sum _{i=1}^{k}u(E_i, T_i)$ + $\sum _{i=T_{e}}^{T_{s} + MTD}u(SE_i, T_i)$ - $u(E_k, T_i)$ = $\sum _{i=1}^{k-1}u(E_i, T_i)$ + $\sum _{i=T_{e}}^{T_{s} + MTD}u(SE_i, T_i)$, where $SE_i$ is the simultaneous event set at the time point $T_i$ in $S$.

\end{definition}

\begin{definition}[Optimized Episode-Weighted Utilization]
	\rm The optimized episode-weighted utilization (optimized \textit{EWU} version 2.0 is denoted as $EWU_{opt'}$) of $\alpha$ = $<$$(E_1$), $(E_2)$, $\ldots$, $(E_{k-1})$, $(E_k)$$>$ w.r.t. $mo(\alpha)$ consists of three parts: 1) the utility of $\alpha$, 2) the remaining utility of $\alpha$ in $T_{e}$, and 3) the utility of all episodes in [$T_{e}+1$, $T_{s}$ + \textit{MTD}]. $EWU_{opt'}$ is defined as \textit{EWU}($\alpha$, $mo(\alpha)$) = $u(\alpha, mo(\alpha))$ + $ru(E_k, T_e)$ + $\sum _{i=T_{e}+1}^{T_{s} + MTD}tu(T_i)$ = $\sum _{i=1}^{k}u(E_i, T_i)$ + $ru(E_k, T_e)$ + $\sum _{i=T_{e}+1}^{T_{s} + MTD}tu(T_i)$, where time point $T_i$ within the satisfied \textit{MTD} interval $[T_e, T_{s}$ + \textit{MTD}]. Thus, for  $\alpha$ in $S$, we have the accumulative  \textit{EWU} w.r.t. $moSet$($\alpha$) = [$mo_{1}, mo_{2}, \ldots, mo_{n}$], such that  \textit{EWU}($\alpha$) = $\sum _{i=1}^{n}EWU(\alpha, mo_{n}(\alpha))$.
\end{definition}

Consider the episode $<$$\{B\} \rightarrow D$$>$ with \textit{MTD} = 2 in Fig. \ref{fig:data}, $EWU_{opt}$($<$$B \rightarrow D$$>$, $[T_2, T_3]$) = \{$u(B$, $T_2)$ + $u(D$, $T_3)$\} + \{$u(\{A, E\}$, $T_3)$ + $u(\{D, E\}$, $T_4)$\} = \{\$10\} + \{\$3 + \$6 + \$7\} + \{\$3 + \$14\} = \$43, while $EWU_{opt'}$($<$$B \rightarrow D$$>$, $[T_2, T_3]$) = \{$u(B$, $T_2)$ + $u(D$, $T_3)$\} + \{$u(E$, $T_3)$ + $u(\{D, E\}$, $T_4)$\} = \{\$10\} + \{\$6 + \$7\} + \{\$3 + \$14\} = \$40. Obviously, $EWU_{opt'}$($<$$B \rightarrow D$$>$, $[T_2, T_3]$)  is more tighter than $EWU_{opt}$($<$$B \rightarrow D$$>$, $[T_2, T_3]$). In general, the larger the average number of events in each time point $T_i$ is, the bigger difference of the results between $EWU_{opt'}$ and $EWU_{opt}$ is.

\begin{table}
	\centering
	\caption{Results of each 1-episode}
	\begin{tabular}{|c|l|c|c|} \hline
		\textbf{Episode}   &   \textbf{\textit{moSet}}   & \textbf{\textit{EWU}}   & \textbf{Utility}      \\ \hline
		$A$ & \{[$T_1,T_1$], [$T_3,T_3$], [$T_6,T_6$]\}  & \$110  &  \$47   \\ \hline
		$B$ & \{[$T_2,T_2$], [$T_5,T_5$], [$T_6,T_6$]\}  & \$105  &    \$49   \\ \hline
		$C$ & \{[$T_1,T_1$], [$T_2,T_2$]\}  & \$90  &    \$49   \\ \hline
		$D$ & \{[$T_1,T_1$], [$T_3,T_3$], [$T_4,T_4$], [$T_6,T_6$]\}  & \$161  &    \$51   \\ \hline
		$E$ & \{[$T_3,T_3$], [$T_4,T_4$]\}  & \$94  &    \$48   \\ \hline
	\end{tabular}
	\label{table:EWU}
\end{table}

Specifically, the \textit{EWU}\footnote{If not otherwise specified, the term \textit{EWU} is an equivalent to $EWU_{opt'}$ in the rest of this paper.} value of $\alpha$ in $S$ is always larger than or equal to the total utility of $\alpha$, as well as the total utility of any of its extension (also called super-episodes) in the search space. The results of each 1-episode in terms of $moSet$, \textit{EWU} and utility are shown in Table \ref{table:EWU}. In the conceptual LS-tree, the upper bound \textit{EWU} can ensure the \textit{local downward closure} property. Based on the above observations, we can use the following filtering strategies. Theorem \ref{theorem:EWU} gives us a way to track patterns that have potential to be HUEs in the subtree. This guarantees to return the exact high-utility episodes. Therefore, it is an important property that our algorithm utilizes the upper bound \textit{EWU} as the pruning strategy, as described below.

\textit{\textbf{Optimized \textit{EWU} strategy}}. When spanning the LS-tree rooted at an episode $\alpha$ as prefix, the UMEpi algorithm spans/explores the search space in a depth-first search way. If the \textit{EWU} of any node/episode $\alpha$ is less than $minUtil \times TU $, any of its child node would not be a final HUE, they can be regarded as low-utility episodes and pruned directly.

The conceptual LS-tree is a prefix-shared tree, and UMEpi extends the lower-level nodes/episodes to the higher-level ones in depth-first way. For each first layer node, they are the 1-HWUEs. Theorem \ref{theorem:EWU} establishes a theoretical basis for generating high-utility $l$-episodes ($l \geq 2$) from the 1-HWUEs. The basic idea of UMEpi is generating all 1-HWUEs by calculating the \textit{EWU} values of all 1-episodes and then directly extending $l$-episodes ($l \geq 2$). Thus, UMEpi searches for the $l$-HWUEs and $l$-HUEs by extending each 1-HWUE as prefix. For any $l$-episode node ($k \geq 2$) in the LS-tree, we can easily search for its 1-extensions using $I$-Concatenation and $S$-Concatenation.

Specifically, Theorem \ref{theorem:EWU} gives us the necessary conditions for computing all prefix-based HWUEs and HUEs. \textit{EWU} is the upper bound on utility of any extension of a node/episode in a subtree. Any subtree with non-HWUE episode $\alpha$ as the root can be pruned since all $\alpha$-prefixed episodes are not HWUEs. Pruning strategies are applied to remove nodes with low upper bound such as $ EWU(\alpha) < minUtil \times TU $.

\subsection{Main Procedure}

To clarify our methodology, we have illustrated the conceptual search space w.r.t. lexicographic sequence tree, the key properties of utility w.r.t. minimum occurrence, and the optimized \textit{EWU} strategy so far. Utilizing the above concepts and technologies, the main procedure of the designed UMEpi algorithm is described in Algorithm \ref{AlgorithmUMEpi}. To summarize, given $S$, the UMEpi algorithm works as follows: (1) In the Phase I, it first scans the event sequence once to obtain the 1-length candidates namely 1-HWUEs. (2) In the Phase II, UMEpi recursively calls the \textit{HUE-Span} procedure to discover the set of $l$-HWUEs and $l$-HUEs having $\alpha$ as prefix. It uses depth-first search for the rest of the mining process. Details are presented below. It takes four parameters as input: 1) a complex event sequence, $S$; 2) a user-specified utility-table, \textit{ptable}; 3) a maximum time duration, \textit{MTD}; and 4) a minimum utility threshold, \textit{minUtil}. The UMEpi algorithm first scans the entire event sequence once to construct the transformed $S'$, and to obtain all 1-length episodes as well as their associated minimal occurrences \textit{moSet} (Line 1). Then it calculates $ EWU(e) $ and the utility of each 1-episode $ e \in E$ in $S$  (Line 2). Besides, the total utility of $S$ can be calculated after these computations. If an 1-episode satisfies $u(e)$ $ \geq$ $ minUtil \times TU$, adds it into the set of valid HUEs  (Lines 5-7).  Here the built \textit{moSet} of all 1-events can be used to calculate the \textit{EWU} value in the later processes, and to avoid projecting the event sequence to a huge amount of sub-sequences.  Thereafter, UMEpi recursively calls the \textit{Span-SimultHUE} and \textit{Span-SerialHUE}  procedures, and finally outputs the complete set of HUEs. It is important to notice that other 1-events that \textit{EWU}$(e)$ $< minUtil \times TU$ also should be explored in the \textit{Span-SimultHUE} and \textit{Span-SerialHUE}  procedures. Because the non-HWUEs may still be the sub-episodes of the final HUEs. In addition, the processing order $\succ$ of events should be kept consistently in UMEpi after Line 3.  For each node/episode, UMEpi does not have to scan the event sequence multiple times to detect the necessary information from the events.

\begin{algorithm}
	\caption{The UMEpi algorithm}
	\label{AlgorithmUMEpi}
	\begin{algorithmic}[1]	
		\REQUIRE \textit{S}; \textit{ptable}; \textit{MTD}; \textit{minUtil}.
		\ENSURE \textit{HUEs}: the complete set of high-utility episodes.     

		\STATE scan original \textit{S} once to construct the transformed $S'$, and to obtain all 1-length episodes as well as their associated minimal occurrences;

		\STATE calculate the \textit{EWU} and real utility values of each 1-episode;

		\STATE sort all 1-episodes in the total order $\succ$;
		
		\FOR{each 1-episode $ e \in S'$ with order $\succ$}
			\IF{ \textit{EWU}$(e) \geq minUtil$ $\times$ $TU$}
			   \IF{ \textit{u}$(e) \geq minUtil$ $\times$ $TU$}
			       \STATE \textit{HUEs} $  \leftarrow $  \textit{HUEs} $ \cup$ $e$;
			   \ENDIF
			
				\STATE call \textbf{Span-SimultHUE}\textbf{($e$, $S'$, \textit{MTD}, \textit{minUtil})};
				
				\STATE call \textbf{Span-SerialHUE}\textbf{($e$, $S'$, \textit{MTD}, \textit{minUtil})};
			\ENDIF
			
		\ENDFOR 		 					    
		
		\STATE \textbf{return} \textit{HUEs}
	\end{algorithmic}
\end{algorithm}

\begin{algorithm}
	\caption{The Span-SimultHUE procedure}
	\label{Simul-search}
	
	\begin{algorithmic}[1]	
		\REQUIRE $\alpha$, $S'$, \textit{MTD}, \textit{minUtil}. 
		\ENSURE \textit{HUEs}: the set of high-utility episodes having $\alpha$ as prefix.    

		\STATE initialize \textit{simultEpiSet} = $\phi$;
				
		\FOR{each $mo(\alpha)$ = $[T_s, T_e] \in$ \textit{moSet}($\alpha$)}
			\STATE \textit{simultEpiSet} $  \leftarrow $ \textit{simultEpiSet} $ \cup$  \{$e$ $|$ simultaneous event $e$ occurs at $T_e$ and $e$ is after the last event in $e$\};
		\ENDFOR

		\FOR{each 1-event/episode $ e \in simultEpiSet$}
		    \STATE  simultaneous episode $\beta  \leftarrow $ \textbf{\textit{Simult-Concatenate($\alpha$, e})}, and calculate  its \textit{moSet}($\beta$);

			\STATE based on $S'$ and \textit{moSet}($\beta$), calculate the overall utility and \textit{EWU} of $\beta$;

			\IF { \textit{EWU}$(\beta) \geq minUtil$ $\times$ $TU$}
			
				\IF {$ u(\beta) \geq minUtil$ $\times$ $TU$}
					\STATE \textit{HUEs} $  \leftarrow $  \textit{HUEs} $ \cup$ $\beta$;
				\ENDIF

				\STATE call \textbf{Span-SimultHUE}\textbf{($\beta$, $S'$, \textit{MTD}, \textit{minUtil})};
				
				\STATE call \textbf{Span-SerialHUE}\textbf{($\beta$, $S'$, \textit{MTD}, \textit{minUtil})};
			\ENDIF
			
		\ENDFOR 		
		
		\STATE \textbf{return} \textit{HUEs}
	\end{algorithmic}	
\end{algorithm}

\renewcommand{\algorithmicrequire}{\textbf{Input:}}
\renewcommand{\algorithmicensure}{\textbf{Output:}}
\begin{algorithm}
	\caption{The Span-SerialHUE procedure}
	\label{Serial-search}
	
	\begin{algorithmic}[1]	
		\REQUIRE $\alpha$, $S'$, \textit{MTD}, \textit{minUtil}. 
		\ENSURE \textit{HUEs}: the set of high-utility episodes having $\alpha$ as prefix.    
		
		\STATE initialize \textit{serialEpiSet} = $\phi$;
		
		\FOR{each $mo(\alpha)$ = $[T_s, T_e] \in$ \textit{moSet}($\alpha$)}
			\FOR{each time point $t$ in [$T_{e}$ + 1, $T_{s}$ + \textit{MTD}]}
				\STATE \textit{serialEpiSet} $  \leftarrow $ \textit{serialEpiSet} $ \cup$   \{$e$ $|$ serial event $e$ occurs at $t$\};
			\ENDFOR 
		\ENDFOR

		\FOR{each 1-event/episode $ e \in serialEpiSet$}

			\STATE get serial episode $\beta  \leftarrow $ \textbf{\textit{Serial-Concatenate($\alpha$, e})}, and calculate  its \textit{moSet}($\beta$);
			
			\STATE based on $S'$ and \textit{moSet}($\beta$), calculate the overall utility and \textit{EWU} of $\beta$;

		\IF { \textit{EWU}($\beta) \geq minUtil$ $\times$ $TU$}		
			\IF {$ u(\beta) \geq minUtil$ $\times$ $TU$}
				\STATE \textit{HUEs} $  \leftarrow $  \textit{HUEs} $ \cup$ $\beta$;
			\ENDIF

			\STATE call \textbf{Span-SimultHUE}\textbf{($\beta$, $S'$, \textit{MTD}, \textit{minUtil})};
			
			\STATE call \textbf{Span-SerialHUE}\textbf{($\beta$, $S'$, \textit{MTD}, \textit{minUtil})};
		\ENDIF
		
		\ENDFOR 		
		
		\STATE \textbf{return} \textit{HUEs}
	\end{algorithmic}	
\end{algorithm}

The \textit{Span-SimultHUE} and \textit{Span-SerialHUE}  procedures are shown in Algorithm \ref{Simul-search} and Algorithm \ref{Serial-search}, respectively. Both of them take as input: 1) a prefix episode $\alpha$, 2) the transformed event sequence $S'$, 3) \textit{MTD}, and 4) $minUtil$. These two procedures are not the same. The \textit{Span-SimultHUE} procedure operates as follows. It first initializes \textit{simultEpiSet} as an empty set (Line 1),  then calculates all simultaneous events of $\alpha$ (that would be expanded simultaneously), based on $S'$ and \textit{moSet}($\alpha$) (Lines 2-3). After obtaining the set of all simultaneous events, it calls the function \textit{Simult-Concatenate($\alpha$, e}) to construct a new simultaneous episode $\beta$, and calculates  its \textit{moSet}($\beta$), as shown in Lines 4-5. Based on $S'$ and \textit{moSet}($\beta$), the overall utility and \textit{EWU} of $\beta$ can be quickly calculated (Line 6). Utilizing the associated \textit{moSet} and transformed event sequence $S'$, therefore, UMEpi can easily obtain the true utility and upper bound of the new generated episode.

After that, the designed \textit{EWU}  pruning strategy is used to determine whether the extensions of $\beta$ would be the HUEs and should be explored (Line 7, using the \textit{EWU} strategy). If the overall utility of $\beta$ is no less than $minUtil \times TU$, this episode will be added into the set of HUEs having $\alpha$ as prefix (Lines 8-9). After filtering the unpromising episode, only the promising episodes with a high \textit{EWU} upper bound would be explored for next extension (Lines 10-11). After all the extensions rooted at  $\alpha$ are performed and determined recursively using the depth-first search mechanism, it finally returns the set of HUEs that having the common prefix $\alpha$ (Line 12). The \textit{Span-SerialHUE} procedure has different operations, as shown in Lines 1-6, but Lines 7-13 as the same as that in the \textit{Span-SimultHUE} procedure. The details for constructing a super-episode (as know as simultaneous episode in \textit{Span-SimultHUE}, and serial episode in \textit{Span-SerialHUE}) of $\alpha$ are not the same. Due the intrinsic sequence-order and complexity, the number of combination of episodes is quite huge. Different from the previous works in frequent episode mining, UMEpi performs an efficient operation for candidate generation.  Besides, both $I$-Concatenation and $S$-Concatenation share a prefix $\alpha$, and they are only allowed to differ in their last event or element.

The UP-Span algorithm utilizes the projected mechanism \cite{pei2001mining} to generate sub-sequences by spanning prefixes, but this is a common trick in SPM to facilitate the prefix-growth process. Note that the projected mechanism is not utilized in our proposed UMEpi algorithm. In UMEpi, both the episode $\alpha$ itself and the minimal occurrence set \textit{moSet}($\alpha$) are stored simultaneously during the mining process. Together with the transformed event sequence $S'$, \textit{moSet}($\alpha$) already provides a complete representation of all sub-sequences w.r.t. $\alpha$. Thus, \textit{moSet}($\alpha$) is sufficient for computing the \textit{EWU} value and utility of $\alpha$ efficiently. In many situations, it suffices to discover the candidate episodes and once the set $moSet$ of their  minimal occurrences is known, all the required candidates and HUEs can be easily generated. Note that UP-Span in unnecessary to adopt the minimal occurrence to calculate the \textit{EWU} value.

%% file: 4_experiment.tex
\section{Experimental Study} \label{sec:experiments}

In this section, we report the performance study of the utility-driven UMEpi algorithm on several real-world datasets. For effectiveness evaluation, UMEpi is compared with the state-of-the-art approaches, UP-Span \cite{wu2013mining} and TSpan \cite{guo2014high}. Note that the UP-Span \cite{wu2013mining} approach does not calculate the \textit{EWU} and real utility of episodes by using the minimal occurrence. As discussed in Section \ref{sec:algorithm}, it suffers from several errors and incorrect mining results. Besides, the TSpan algorithm \cite{guo2014high} also fails to extract the correct HUEs. Thus, the execution time and memory cost of UP-Span and TSpan are not suitable to evaluate the efficiency of the developed UMEpi algorithm.  We perform the experiments to make an assessment of UMEpi's efficiency, in terms of:  1) How efficient is UMEpi with a variety of datasets and parameters? 2) How powerful are the different optimized \textit{EWU} strategies?  3) What impact would be caused by different processing order of 1-HWUEs?

To answer these questions, three variants of UMEpi are conducted in our experiments: UMEpi$_{baseline}$ denotes the proposed algorithm adopt the original \textit{EWU} as defined in TSpan, UMEpi adopts the optimized \textit{EWU} version 1.0 (\textit{EWU}$_{opt}$), and UMEpi+ utilizes the most tightest upper bound  \textit{EWU}$_{opt'}$.

\subsection{Data Description and Experimental Setup}

\textbf{Datasets}. Our experiments are conducted on four real-world datasets, including \textit{retail}\footnote{\url{http://fimi.ua.ac.be/data/}}, \textit{BMS}\footnote{\url{http://fimi.ua.ac.be/data/}}, \textit{foodmart}\footnote{\url{http://msdn.microsoft.com/enus/library/aa217032(v=sql.80).asp}}, \textit{chainstore}\footnote{\url{http://cucis.ece.northwestern.edu/projects/DMS/MineBench.html}}. For FIM and HUIM, these datasets are the general transaction data. For FEM and HUEM, they can be viewed as a single complex sequence when each item is regarded as an event and each transaction is viewed as a simultaneous event set.  The characteristics of these four datasets are described below.

\begin{itemize}

\item  \emph{\textbf{retail}}: it is collected from real-life retail data. There are totally 88,162 transactions with 16,470 distinct items.  And an average transaction length is 10.3 items. 

\item  \emph{\textbf{BMS}}: this click-stream data is collected from an e-commerce, and contains 59,601 sub-sequences with 497 distinct items. The average length of sub-sequences in BMS is 2.42 items with a standard deviation of 3.22.

\item \emph{\textbf{foodmart}}: a dataset of customer transactions from an anonymous chain store, provided by Microsoft SQL Server. There are totally 21,556 transactions and 1,559 distinct items.

\item  \emph{\textbf{chainstore}}: this dataset is obtained and transformed from NU-Mine Bench. As a real-life dataset of customer transactions from a retail store, chainstore contains 1,112,949 transactions with 46,086 distinct items/events. Besides, the average transaction length ups to 7.3 items.
\end{itemize}

More details of these datasets can be referred to the SPMF website\footnote{\url{http://www.philippe-fournier-viger.com/spmf/index.php?link=datasets.php}}, and all of them are published and available to researchers. Both foodmart and chainstore contain the embedding occur quantity and unit utility of each item, while retail and BMS only contain the information of items/events. Similar to the previous studies \cite{tseng2013efficient,liu2012mining}, we use a simulation method to randomly generate the internal and external utilities in retail and BMS: 1) generate the occur quantity (in range of 1 to 6) for each item in every transaction; 2) set the unit utility for each item (in range of 1 to 1000 by using a log-normal distribution).

\textbf{Evaluation platform}. In our experiments, all the algorithms are written in Java language. The source C++ code of TSpan is provided by its authors, and the implementation of UP-Span is available at SPMF website.  The experiments are conducted on a personal ThinkPad T470p computer with an Intel(R) Core(TM) i7-7700HQ CPU @ 2.80 GHz 2.81 GHz, 32 GB of RAM, and with the 64-bit Microsoft Windows 10 operating system. The Maximum JVM memory is set to 8 GB of RAM.

\textbf{Parameter settings}. Note that the default size of each dataset is 1000 (the 1000 events/transactions in front are selected),  and the maximal duration time of the desired high-utility episode is fixed set as \textit{MTD} = 4 when varying \textit{minUtils}. For each test dataset with a fixed size, when \textit{minUtil} is fixed set, \textit{MTD} is varied set from 1 to 6.  We run each method three times and report the average results (i.e., execution time, memory consumption). When the runtime exceeds 100,000 seconds or the algorithm is out of memory, we assume that there is no result of runtime and memory consumption. Then, the result of related patterns is marked as ``-".

\subsection{Effectiveness Evaluation}

We first study the effectiveness of UMEpi. First, we use the running example as a case study to evaluate the final patterns that discovered by the pioneer UP-Span \cite{wu2013mining}, the state-of-the-art TSpan \cite{guo2014high}, and the proposed UMEpi algorithm. The mining results from event sequence in Fig. \ref{fig:data} are plotted in Table \ref{table:example}, with a fixed \textit{MTD} = 3. Note that \textit{test}$_1 $ to \textit{test}$_6$ is related to \textit{minUtil}: 30\%, \textit{minUtil}: 35\%, \textit{minUtil}: 40\%, \textit{minUtil}: 45\%, \textit{minUtil}: 50\%, and \textit{minUtil}: 55\%, respectively. Obviously, both UP-Span and TSpan fail to extract the complete true high-utility episodes from event sequence. Although the number of HUEs extracted by UP-Span is close to the final HUEs returned by UMEpi, the utilities of some HUEs are incorrect.

\begin{table}[htb]
	\fontsize{5.5pt}{9pt}\selectfont
	\centering
	\caption{Discovered HUEs by three HUEM algorithms}
	\label{table:example}
	\begin{tabular}{|c|c|llllll|}
		\hline\hline
		\multirow{2}*{\textbf{Dataset}}&
		\multirow{2}*{\textbf{Algorithm}}
		&\multicolumn{6}{c|}{\textbf{\# HUEs when varying \textit{minUtil} with \textit{MTD} = 3}}\\
		\cline{3-8}
		&&\textit{test}$_1 $ & \textit{test}$_2$ & \textit{test}$_3$ & \textit{test}$_4$  &  \textit{test}$_5$  &  \textit{test}$_6$ \\ \hline

		&\textbf{UP-Span} &	 863  &	 647  &	  385  &	218   &	 113  &	 53 	 \\
		
		\textbf{Example}  &  \textbf{TSpan}   &	 14   &	 11  &	 6  &	5   &	1   &	0 	 \\
		(Fig. \ref{fig:data}) &\textbf{UMEpi}  &	891   &	717  &	476   &  301   &	 174   &	69 	 \\

		\hline
		
		\hline
\hline
\end{tabular}
\end{table}

In the following experiments, both UP-Span and TSpan are not compared since both of them fail to solve the addressed problem for mining high-utility episodes.  To further evaluate the effectiveness of UMEpi, we plot the results of candidate episodes and the final HUEs under various parameter settings in Table \ref{table:patterns1}. Note that \#HUEs denotes the number of final HUEs discovered by three variants of UMEpi (UMEpi$_{baseline}$, UMEpi, and UMEpi+), and the number of visited candidates/nodes of three variants is denoted as \#$N_1$, \#$N_2$ and \#$N_3$, respectively. In Table \ref{table:patterns1}, $\delta$ represents \textit{minUtil}, and its detailed value is shown in Fig. \ref{fig:Runtime}(a) to (d).  In Table \ref{table:patterns2}, the fixed \textit{minUtil} of each dataset is shown in Fig. \ref{fig:Runtime}(e) to (h).  If we look at it in another light, the number of candidate episodes can also be used to assess the effects of the adopted pruning strategies.

\begin{table}[htb]
	\fontsize{4.8pt}{10pt}\selectfont
	\centering
	\caption{\# patterns under varying \textit{minUtil} with fixed \textit{MTD}}
	\label{table:patterns1}
	\begin{tabular}{|c|c|llllll|}
		\hline\hline
		\multirow{2}*{\textbf{Dataset}}&
		\multirow{2}*{\textbf{\# Patterns}}
		&\multicolumn{6}{c|}{\textbf{\# patterns under different thresholds}}\\
		\cline{3-8}
		&&$ \delta_1 $ & $ \delta_2 $ & $ \delta_3 $ & $ \delta_4 $ &  $ \delta_5 $  & $ \delta_6 $ \\ \hline

 &\textbf{\#${N_1}$} &	77,758  &	67,478  &	59,204  &	52,474  &	46,912  &	42,206	 \\
	 
\textbf{retail}  &  \textbf{\#${N_2}$}   &	 69,936  &	60,580  &	53,406  &	47,624  &	42,538  &	38,320 	 \\
 &\textbf{\#${N_3}$}  &	 57,210  &	50,556  &	45,150  &	40,372  &	36,430  &	33,042	 \\
&  \#HUEs   &	640  &	541  &	462  &	391  &	345  &	299	 \\
\hline

&\textbf{\#${N_1}$} &	9,472  &	8,746  &	7,760  &	6,294  &	3,450  &	2,388 	 \\

\textbf{BMS}  &  \textbf{\#${N_2}$}   &	 8,556  &	2,464  &	2,432  &	2,408  &	2,364  &	2,338 	 \\
&\textbf{\#${N_3}$}  &	274  &	256  &	250  &	240  &	236  &	226	 \\
&  \#HUEs   &	 0  &	0  &	0  &	0  &	0  &	0 	 \\
\hline

&\textbf{\#${N_1}$} & -	  &  -	 & 	12,804,602 &	6,117,054 &	2,733,572 &	1,157,568 	 \\

\textbf{foodmart}  &  \textbf{\#${N_2}$}   &	5,550,420 &	2,542,604 &	1,092,060 &	436,324 &	164,712 &	58,954 	 \\
&\textbf{\#${N_3}$}  &	 369,006 &	162,018 &	67,648 &	27,080 &	11,248 &	5,322 	 \\
&  \#HUEs   &	34,064 &	13,313 &	4,826 &	1,582 &	476 &	122 	 \\
\hline

&\textbf{\#${N_1}$} &	153,856  &	43,454  &	24,492  &	1,972  &	1,524  &	1,520	 \\

\textbf{chainstore}  &  \textbf{\#${N_2}$}   & 934  &	928  &	918  &	916  &	912	  & 906 	 \\
&\textbf{\#${N_3}$}  &	658  &	656  &	652  &	650  &	644  &	640	 \\
&  \#HUEs   &	3  &	3  &	3  &	3  &	3  &	3	 \\
\hline

		\hline
\hline
	\end{tabular}
\end{table}

\begin{table}[htb]
	\fontsize{5.5pt}{9pt}\selectfont
	\centering
	\caption{\# patterns under varying \textit{MTD} with fixed \textit{minUtil}}
	\label{table:patterns2}
	\begin{tabular}{|c|c|llllll|}
		\hline\hline
		\multirow{2}*{\textbf{Dataset}}&
		\multirow{2}*{\textbf{\# Patterns}}
		&\multicolumn{6}{c|}{\textbf{\# patterns under different thresholds}}\\
		\cline{3-8}
		&& 1 & 2 & 3 & 4 &  5  & 6 \\ \hline

		&\textbf{\#${N_1}$} &	 2,084 &	7,180 &	22,896 &	67,478 &	178,830	 & -	 \\
		
		\textbf{retail}  &  \textbf{\#${N_2}$}   &	 1,806 &	6,258 &	20,386 &	60,580 &	162,090 &	451,756 	 \\
		&\textbf{\#${N_3}$}  &	1,426 &	5,022 &	16,836 &	50,556 &	137,056	 & 389,726	 \\
		&  \#HUEs   &	31 &	80 &	211 &	541 &	1,284 &	3,044	 \\
		\hline

		&\textbf{\#${N_1}$} &	1,478 &	2,564 &	4,082 &	9,472	 & 	- & -	 \\
		
		\textbf{BMS}  &  \textbf{\#${N_2}$}   &	274	 & 2,250 &	2,320 &	8,556	 & 9,068  & -	 \\
		&\textbf{\#${N_3}$}  &	 118 &	170 &	214 &	274 &	334 &	524 	 \\
		&  \#HUEs   &	 0 &	0 &	0 &	0 &	0 &	0 	 \\
		\hline

		&\textbf{\#${N_1}$} &	726 &	1,552 &	2,240 &	12,804,602	 & -	 & 	- \\
		
		\textbf{foodmart}  &  \textbf{\#${N_2}$}   &	724	 & 1,528 &	2,176 &	1,092,060 & -	 & 	- \\
		&\textbf{\#${N_3}$}  &	 722 &	1,510 &	2,096 &	67,648	 & 	- & -	 \\
		&  \#HUEs   &	0  &	0 &	0 &	4,826 	 & -	 & -	 \\
		\hline

		&\textbf{\#${N_1}$} &	178 &	372	 & 752 &	24,492 &	44,018 	 & 	- \\
		
		\textbf{chainstore}  &  \textbf{\#${N_2}$}   &	  172 &	362	 & 586 &	918 &	1,384	 & -	 \\
		&\textbf{\#${N_3}$}  &	 154 &	270	 & 420 &	652 &	940	 & -	 \\
		&  \#HUEs   &	 2 &	2 &	3 &	3 &	3	 & -	 \\
		\hline

		\hline
		\hline
	\end{tabular}
\end{table}

As shown in Table \ref{table:patterns1}, it can be clearly observed that the number of HUEs is always quite less than that of the intermediate candidates, and the number of candidates generated by the designed three variants always has the following relationship: $ \#N_1 \geq  \#N_2 \geq \#N_3$, under various \textit{MTD} and \textit{minUtil} thresholds.  These results are reasonable, the reason is that there are a huge number of candidates that are determined for mining HUEs, and three variants of UMEpi adopt different optimized \textit{EWU} strategies. To be more specific, a more tighter upper bound of \textit{EWU} may lead to a more smaller search space and more less candidates. For example, in the foodmart dataset with \textit{minUtil}: 0.88\%, the number of patterns is $\#N_1$: 12,804,602, $\#N_2$: 1,092,060, $\#N_3$: 67,648, and \#HUEs: 4,826, respectively. Consider the utility measure and \textit{MTD} constraint, we normally get the above observations also in other datasets. Therefore, it can clearly show the effect of the optimized \textit{EWU} pruning strategies.

Specifically, both the \textit{MTD} and \textit{minUtil} affect the results of visited candidates and HUEs, as shown form \#$N_1$ to \#$N_3$ and \#HUEs on each dataset. In general, the number of HUEs are always much smaller than that of candidates. The reason is that the actual search space of the HUEM task is huge although the number of actual high-utility episodes is rare. For example, it may return up to hundreds of high-utility episodes in the running example, which only has six event sets (time points) as presented at Fig. \ref{fig:data}.  When \textit{minUtil} is fixed on one specific dataset, the larger \textit{MTD} is, the larger the number of discovered patterns is. Consider the retail dataset, when \textit{MTD} is set as 1, $\#N_3$ = 1,426, \#HUEs = 31; when \textit{MTD} is changed to 6, $\#N_3$ = 389,726, \#HUEs = 3,044. Therefore, in general, the \textit{MTD} parameter cannot be set too large. Otherwise, it may easily lead to combinatorial explosion. Besides, it indicates that we need to develop an efficient HUEM algorithm with an acceptable mining efficiency to extract high-utility episodes from complex event sequence.

\subsection{Runtime Analysis}

The runtime analysis is given in this subsection. The ``runtime" indicates the average running time of each variant of the UMEpi algorithm, by varying \textit{minUtil} or \textit{MTD}. In Fig. \ref{fig:Runtime}, (a) to (d) are tested with a fixed \textit{MTD} (= 4) and varied \textit{minUtils}; (e) to (h) are tested with a fixed \textit{minUtil}  and varied \textit{MTDs} (from 1 to 6).  Note that the scale of the runtime is second.

\begin{figure*}[htbp]
	\centering 
	\includegraphics[trim=20 20 15 10,clip,scale=0.57]{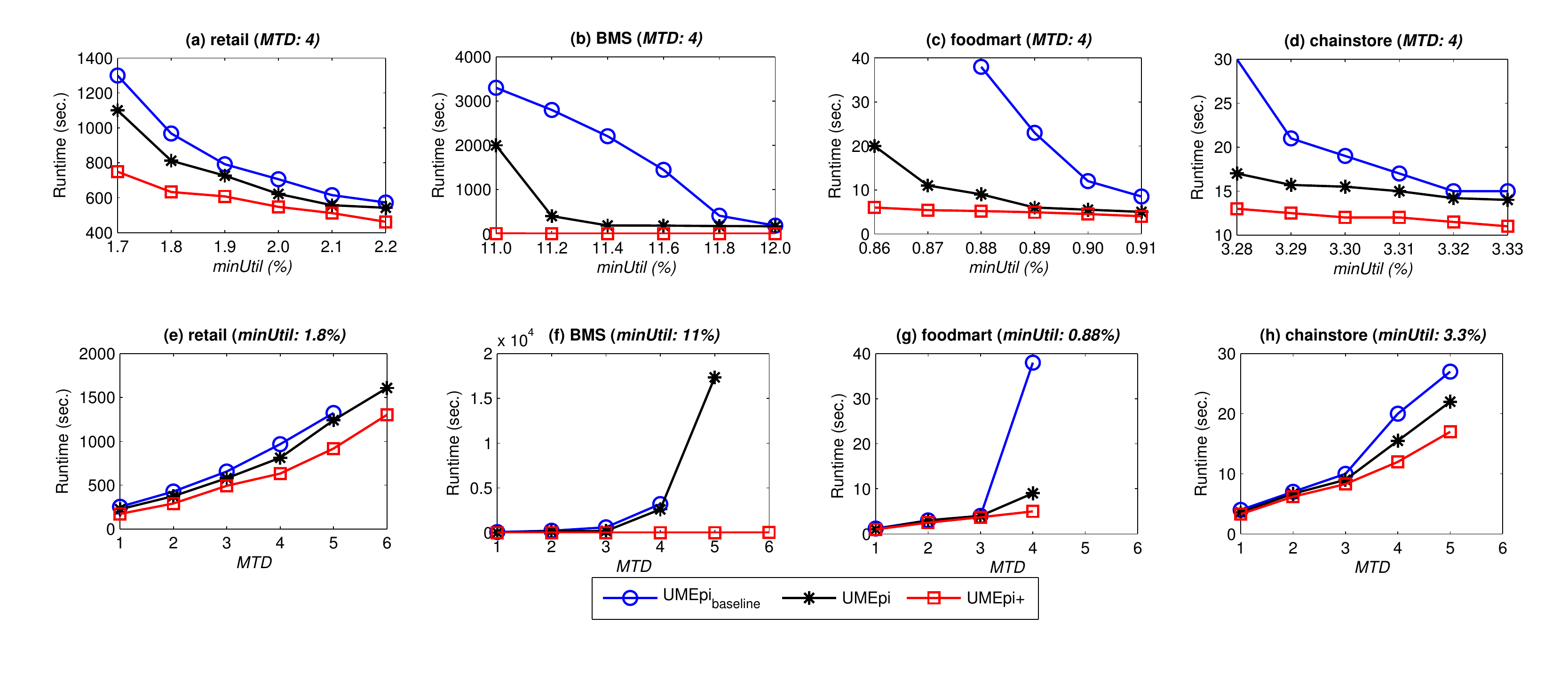}
	\captionsetup{justification=centering}
	\caption{Runtime under various parameters (\textit{MTD} and \textit{minUtil}).}
	\label{fig:Runtime}	
\end{figure*}

It can be seen that in each sub-figure, the total execution time of each compared algorithm is highly related to the number of final HUEs. Note that if the two user-specified  parameters are poorly chosen, like a large \textit{MTD} or a low \textit{minUtil} in Fig. \ref{fig:Runtime}, a relatively large number of episodes seems to become high utility, and thus the running time would increase dramatically. Besides, the differences of runtime among three compared algorithms are largely related to the number of their generated candidate episodes. For example, in the case of BMS as shown in Fig. \ref{fig:Runtime}(b) and Fig. \ref{fig:Runtime}(f), we can obviously observe the difference of the execution time  between UMEpi$_{baseline}$ and UMEpi+ and its trend. When $minUtil$ is set to 11\% on BMS dataset, the runtime of UMEpi+ is always close to 7 seconds, while  UMEpi$_{baseline}$ approximately has its processing time as 3,200 seconds.

An analysis of the effects of changing \textit{minUtil} (Fig. \ref{fig:Runtime}(a), Fig. \ref{fig:Runtime}(b), Fig. \ref{fig:Runtime}(c), and Fig. \ref{fig:Runtime}(d)) shows that a poorly chosen \textit{minUtil} might result in a long-running time. Besides, the \textit{MTD}  parameter gives a fine-grained control of the number of interesting results. When setting a large \textit{MTD}, the total numbers of episodes and HUEs increase. Thus, a longer running time is required to determine these patterns.  Note that, especially for the dataset where average transaction length is large, a low \textit{minUtil} or a large \textit{MTD} can dramatically increase the running time. This is due to a large number of (candidate) episodes being generated and determined. From the experimental results, there seems to be a cutoff point that separates high utility episodes from low utility ones. Small changes around the \textit{MTD} may have a noticeable effect on the runtime.

\subsection{Memory Cost Analysis}

Next, we analyze the memory cost performance of the proposed algorithm. In this set of experiments, we evaluate the effect on memory cost of different \textit{EWU} strategies in UMEpi for discovering high-utility episodes. Fig. \ref{fig:Memory} shows the maximal memory cost when we vary the \textit{minUtil} or \textit{MTD} with the fixed size of the target dataset. Note that in all datasets, we use Java API to count the maximal memory consumption for fair comparison. In Fig. \ref{fig:Memory}, the missing value means the runtime exceeds 100,000 seconds or the algorithm is out of memory.

\begin{figure*}[htbp]
	\centering 
	\includegraphics[trim=25 25 30 10,clip,scale=0.57]{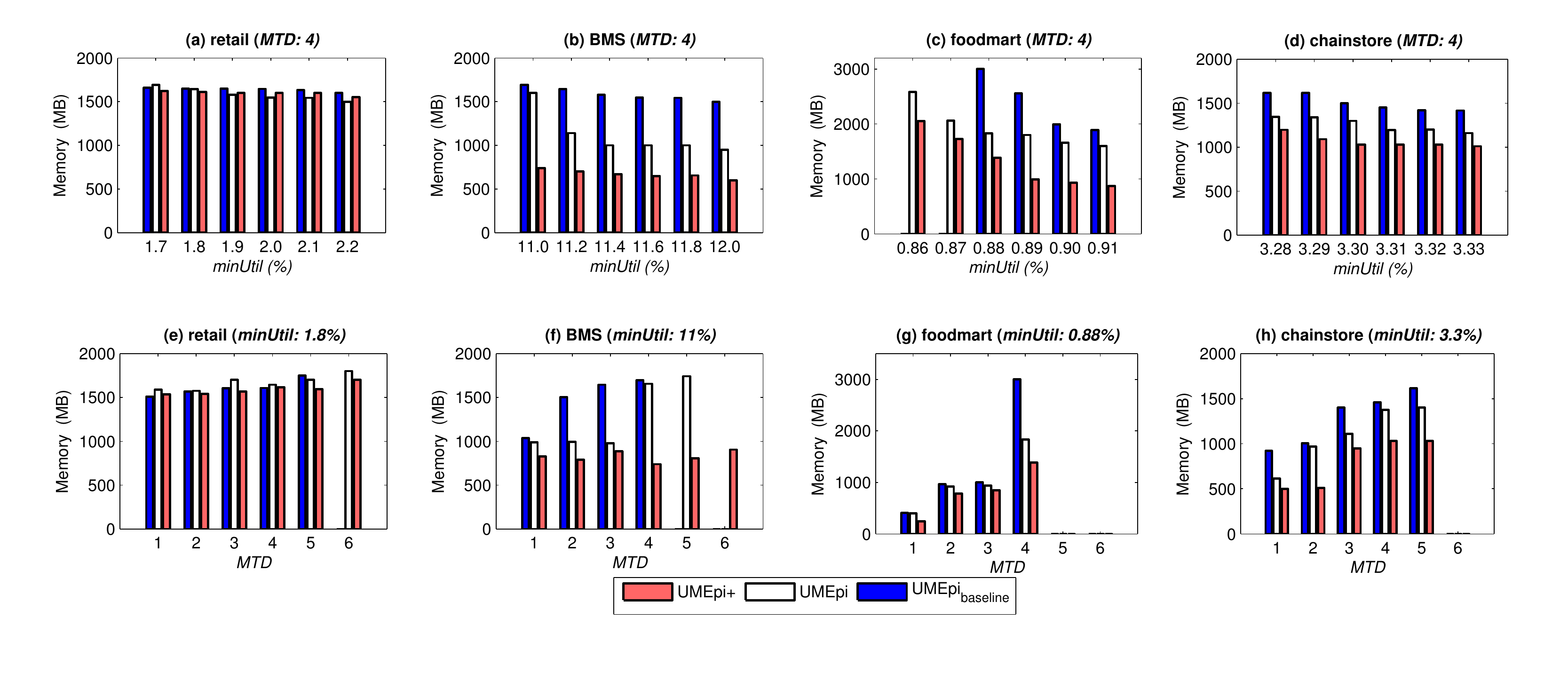}
	\captionsetup{justification=centering}
	\caption{Memory cost under various parameters (\textit{MTD} and \textit{minUtil}).}
	\label{fig:Memory}	
\end{figure*}

In foodmart dataset, all three variants have similar memory consumption no matter varying \textit{minUtil} or varying \textit{MTD}. In other datasets, UMEpi consumes a much smaller memory than that of UMEpi$_{baseline}$ and UMEpi+ performs the best under all parameter settings. Intuitively, the results of memory consumption of three UMEpi variants are directly related to the number of candidate patterns which are shown in Table \ref{table:patterns1}. We can also obtain the following relationship as: $ \#N_1 \geq  \#N_2 \geq \#N_3$, under various \textit{MTD} and \textit{minUtil} thresholds

Note that UMEpi+ consumes a much smaller memory space in each dataset than that of UMEpi$_{baseline}$ and UMEpi. Utilizing a more tighter upper bound of \textit{EWU}, the improved UMEpi+ algorithm can save the cost of candidate generation and speed up processing when spanning the LS-tree. As shown in Table \ref{table:patterns1}, the intermediate candidates can be reduced significantly, thereby reducing memory access and storage costs. For example, in the case of Fig. \ref{fig:Memory}(b) and Fig. \ref{fig:Memory}(f), it is clear that the difference of the maximal memory consumption of UMEpi$_{baseline}$, UMEpi and UMEpi+ is obvious. When \textit{minUtil} is set to 11.8\% on BMS dataset,  UMEpi$_{baseline}$ and UMEpi respectively consumes 1,000 MB and 990 MB, while UMEpi+ approximately consumes 676 MB, which is quite less than the other ones. This trend also can be observed from the other datasets. Therefore, we can conclude that the proposed new \textit{EWU} upper bound plays an active role in pruning the search space of the utility-driven UMEpi algorithm.

\subsection{Processing Order of Events}
In general, the processing order of a data mining algorithm may influence the mining performance. What processing order is more suitable for the proposed algorithm? To assess the influence caused by different processing orders, we compare the runtime and memory consumption of UMEpi+ using different processing orders. Totally four types of processing orders are tested on retail under the same parameter settings (\textit{MTD} = 4 and \textit{minUtil} is varied from 1.7\% to 2.2\%). The UMEpi$ _{occ} $, UMEpi$ _{lexi} $, UMEpi$ _{ewuas} $, and UMEpi$ _{ewude} $ refers to the occur order, the lexicographic order, the \textit{EWU} ascending order, and the \textit{EWU} descending order, respectively. As we can see from the experimental results in Fig. \ref{fig:OrderOfItems}(a), UMEpi$ _{lexi} $ requires the longest runtime, while UMEpi$ _{ewuas} $ is the fastest among the four compared orders. Consider the memory consumption as shown in Fig. \ref{fig:OrderOfItems}(b), obviously, UMEpi$ _{lexi} $ and UMEpi$ _{ewuas} $ have the similar performance on memory cost. The UMEpi$ _{ewude} $ and UMEpi$ _{occ} $ consume similar memory. To summary, UMEpi$ _{ewuas} $ always requires less memory consumption than that of other three orders, including UMEpi$ _{lexi} $, UMEpi$ _{ewude} $, and UMEpi$ _{occ} $. Taken together, the adopted \textit{EWU} ascending order of events (1-episodes) (UMEpi$ _{ewuas} $) can lead to the best performance in terms of runtime and memory usage. In addition, any processing order in UMEpi does not have an effect on the mining results  since the number of candidates and final HUEs in three variants are the same. Details of the number of pattens are not shown here due to the space limit.

\begin{figure}[!htbp]
	\centering
	\includegraphics[trim=120 230 360 0,clip,scale=0.53]{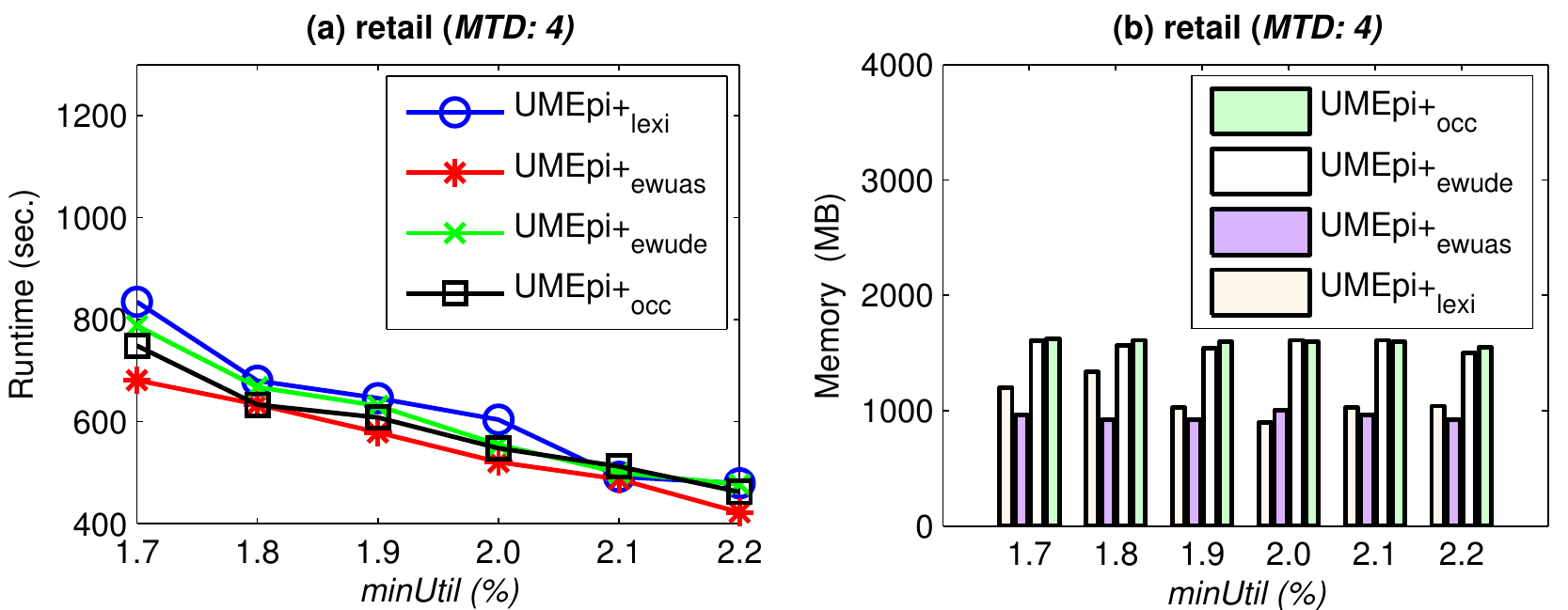}
	\caption{Effect of processing order.}
	\label{fig:OrderOfItems}
\end{figure}

\subsection{Scalability Test} 

\begin{figure*}[htbp]
	\centering 
	\includegraphics[trim=110 245 160 0,clip,scale=0.75]{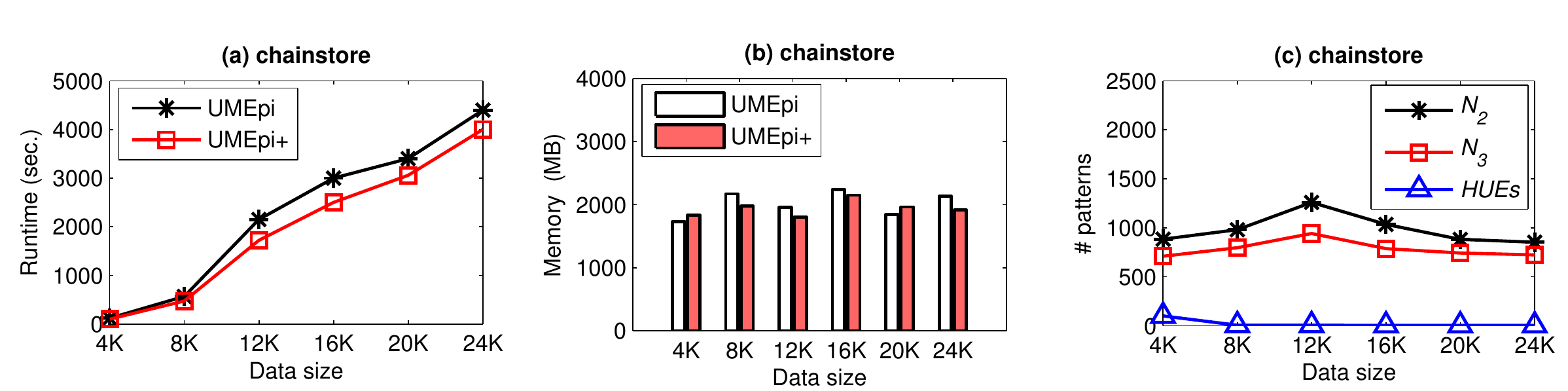}
	\captionsetup{justification=centering}
	\caption{Scalability test with different data size.}
	\label{fig:Scalability}	
\end{figure*}

The computational efficiency problem for HUEM might be more likely happened in long event sequences. Results in Fig. \ref{fig:Scalability} showed the scalability of how UMEpi performs with different numbers of events in a complex sequence. In Fig. \ref{fig:Scalability}, the parameter settings are \textit{MTD} = 4, \textit{minUtil} = 2\%, and dataset size is changed from 4,000 to 24,000 simultaneous event sets.  It shows that the execution time is linear  with respect to the number of events $|$$K$$|$ in chainstore, while the memory cost, the number of intermediate candidates and final HUEs are not that. We have also noticed that execution time of UMEpi and UMEpi+ increases rather gradually, and UMEpi+ always faster than UMEpi. With the fixed \textit{MTD} and \textit{minUtil}, we refer to the candidates and final HUEs whose numbers do not dramatically change, but more execution time is required due to the large sequence size for discovering high-utility episodes.


%% file: 6_conclusion.tex
\section{Conclusions}
\label{sec:conclusion}

In this paper, we have presented a generic utility-driven episode mining framework named UMEpi for discovering high-utility episodes from a complex event sequence. A generic concept namely episode-weighted utilization (\textit{EWU}) is defined and  the optimization strategies are further introduced to reduce this upper bound on utility of episodes. To the best of our studies, UMEpi is the first algorithm that can successfully solve the problem of discovering  high-utility episodes.  Based on the optimized \textit{EWU} concept, UMEpi applied the powerful pruning strategies which utilize the downward closure property of \textit{EWU} to prune the search space. Moreover, UMEpi can directly discover high-utility episodes by avoiding performing costly operations of unpromising candidates. The extensive performance on several real-world datasets demonstrates the effectiveness of UMEpi compared to the existing approaches. Furthermore, UMEpi has good mining performance in terms of efficiency and scalability.

\section*{Acknowledgment}


We would like to thank  Dr. Guangming Guo for sharing the original C++ code of the TSpan algorithm.  This research was partially supported by the Shenzhen Technical Project under Grant No. JCYJ 20170307151733005 and No. KQJSCX 20170726103424709, and a grant from CSC (China Scholarship Council) Program. 


